\documentclass{article}
\usepackage{graphicx} 
\usepackage{amsmath,amssymb,amsthm}
\usepackage{braket}
\usepackage{xcolor}
\usepackage[top=1in, bottom=1in, left=1.3in, right=1.3in]{geometry} 

\usepackage{subcaption}
\usepackage{graphicx}
\usepackage{url}
\usepackage[hidelinks]{hyperref}
\usepackage[noabbrev,nameinlink,capitalize]{cleveref}

\usepackage{titlesec}
\titleformat{\paragraph}
{\normalfont\normalsize\bfseries}{\theparagraph}{1em}{}
\titlespacing*{\paragraph}
{0pt}{3.25ex plus 1ex minus .2ex}{1.5ex plus .2ex}

\usepackage{tikz}
\usepackage{quantikz}
\usepackage{caption}

\usepackage[maxbibnames=15, url=false, style=alphabetic, giveninits=true]{biblatex}
\newcommand{\eps}{\varepsilon}
\newtheorem{theorem}{Theorem}

\usepackage{booktabs}
\usepackage{threeparttable}
\usepackage{multirow}
\usepackage{array}

\title{Probability distribution reconstruction using circuit cutting applied to a variational classifier}
\author{Niels M. P. Neumann \and Carlos M. R. Rocha \and Jasper Verbree \and Marc van Vliet\footnote{Authors ordered alphabetically}}
\date{Department of Applied Cryptography and Quantum Applications, \\ The Netherlands Organisation for Applied Scientific Research, \\ 2595DA The Hague, The Netherlands}

\begin{document}

\maketitle

\begin{abstract}
Significant efforts are being spent on building a quantum computer. 
At the same time, developments in quantum software are rapidly progressing. 
Insufficient quantum resources often are the problem when running quantum algorithms. 
New techniques can aid in using smaller quantum computers to run larger quantum algorithms. 
One of these techniques is circuit cutting. 
With this method, a circuit is broken into multiple pieces, each of which is run on quantum hardware independently and then recombined to obtain the overall answer. 
These circuit cutting techniques require additional circuit evaluations, which can form a bottleneck for algorithms requiring many interactions. 
This work explores the potential of circuit cutting techniques precisely in this regime of many interactions. 
We consider two different models, a standard method based on expectation values, and a novel method based on probability distribution reconstruction. 
Next, we compare two different training methods we call \textit{train-then-cut} and \textit{cut-then-train} and show that in practice, cut-then-train still produces good results. 
This observation brings closer the practical applicability of circuit cutting techniques, as a \textit{train-then-cut} strategy is often infeasible. 
We conclude by implementing a cut and uncut circuit and find that circuit cutting helps achieve higher fidelity results. 
\end{abstract}

\section{Introduction}
Near-term quantum devices suffer from a limited number of good qubits. 
Practical usage of near-term quantum devices requires ways to utilize the available qubits optimally. 
Different techniques exist to optimally utilize the quantum resources, such as error-correcting codes~\cite{Shor_error:1995,McClean2020}, error-mitigation techniques~\cite{Geller2021,Bravyi:2021} and improved compiling of quantum circuits~\cite{Hashim:2021}.

In 2020, \citeauthor{Peng:2020} proposed a novel technique to deal with quantum devices with a limited number of qubits by showing how large quantum circuits can be approximated using small quantum computers~\cite{Peng:2020}. 
Their method identifies clusters of at most $d$ qubits, where for each cluster, at most $k$ qubits interact with another cluster. 
The overhead of this method is exponential in~$k$. 

Since then, multiple improvements and variants of the original algorithm have been proposed. 
Multiple works improved the classical overhead, for instance, by using randomized intermediate measurements~\cite{Lowe:2023,Harada:2024}, by allowing classical communication~\cite{Piveateau:2024}, by reducing the number of CNOT gates used~\cite{Ufrecht:2023}, or by only considering states of high amplitude~\cite{Vedran:2023}. 
Other works gave efficient cutting protocols for specific quantum operations, such as two-qubit unitaries~\cite{Mitarai:2021,Mitarai2021overheadsimulating,Ufrecht:2024}.
\citeauthor{Schmitt:2025} showed that cutting gates simultaneously is advantageous over cutting them separately~\cite{Schmitt:2025}. 

Other works implemented circuit cutting methods to run variational quantum eigensolvers~\cite{Fujii:2022} or the quantum approximate optimization algorithm (QAOA)~\cite{Bechtold:2023,Meulen:2025}.
Additionally, circuit cutting methods can help obfuscate the quantum operations~\cite{Typaldos:2024}.

Circuit cutting methods help to decompose a large quantum circuit into smaller ones. 
These smaller quantum circuits can have smaller depth or width. 
With shorter depth or smaller width, the circuits tend to be less susceptible to noise, resulting in higher-fidelity results. 
Various implementations show that applying circuit cutting improves the overall fidelity of hardware implementations~\cite{Ufrecht:2023,Ying:2023,Singh:2024}.

With most circuit cutting results, the subcircuits are recombined to approximate the expectation value of the original circuit. 
Hoeffding's inequality promises that these methods succeed in this task with high probability~\cite{Hoeffding:1963}.
However, some algorithms, such as QAOA, variational algorithms, and Grover's algorithm with multiple targets~\cite{Grover:1996}, require an output state with large amplitude, instead of the expectation value of the circuit. 
Traditional circuit cutting techniques are thus insufficient for such applications. 

We extend this area of research in two ways. 
First, we propose a method to reconstruct the probability distribution of a quantum circuit, overcomming the limitation of traditional circuit cutting techniques. 

Second, we introduce and compare two training strategies, called \textit{fit\_then\_cut} and \textit{cut\_then\_fit}.
This extension addresses the training part of variational algorithms. 
Training a variational algorithm in a circuit cutting setting can be extremely resource intensive. 
We show, by means of simulation, and hardware implementations, that circuit cutting has potential for variational algorithms, but that the training cost can be a limiting factor for large scale problems. 

The second extension considers a quantum variational classifier~\cite{Benedetti:2019,TNOVCBLOG:2023,TNOVCSOFTWARE:2025} with circuit cutting applied to it.
We will introduce circuit cuts in the training and/or validation phase of the algorithm to test its performance. 
As an added benefit, our work compares the \textit{modulo} model and the \textit{parity} model, showing that the label encoding affects the overall performance, a feature already noticed in earlier work~\cite{Meyer:2023,TNOVCBLOG:2023}. 

The remainder of this work is structured as follows: 
\cref{sec:background} gives background information on circuit cutting.
This section also introduces our first extension: the reconstruction of the probability distribution using circuit cutting methods. 
Next, \cref{sec:VC} describes the use case considered, the corresponding quantum circuit, and the framework used to run the experiments. 
\cref{sec:Results} shows the results of the executed experiments.
These results include simulated and hardware results and a comparison between different label encodings and different training methods. 
This work concludes in \cref{sec:Discussion} with a discussion of the results and directions for future work.

\section{Circuit cutting}\label{sec:background}
This section provides background information on circuit cutting, as well as our proposed method to reconstruct the probability distribution. 

\subsection{Background on circuit cutting}
The original approach to circuit cutting uses an $n$-qubit device and classical computations to mimic an $n+k$-qubit algorithm. 
Subsequent works showed how to subdivide a quantum circuit in subcircuits, even if $k$ is in the order of $n$.
Circuit cutting techniques can apply gate cutting and wire cutting. 
Gate cutting replaces a multi-qubit gate by single-qubit gates and measurements, whereas wire cutting measures a qubit and reinitializes it. 
\cref{fig:circuit_cutting_example} graphically shows the idea behind circuit cutting. 
\begin{figure}
    \centering
    \includegraphics[width=\linewidth]{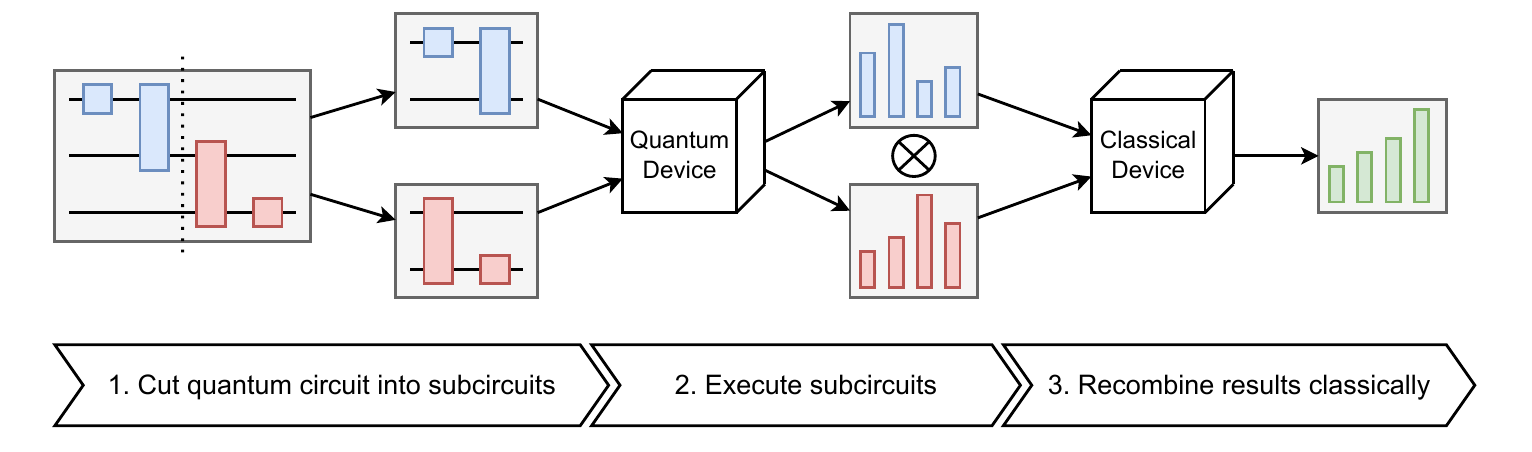}
    \caption{A graphical explanation of circuit cutting. 
    In this example, the second qubit is cut and the resulting circuit consists of two smaller quantum circuits.
    Figure used from~\cite{Meulen:2025}.}
    \label{fig:circuit_cutting_example}
\end{figure}

These gate replacements and qubit measurements can result in the circuit being decomposable into independent blocks.
Each block tends to use fewer gates and/or qubits than the original circuit. 
Running the independent blocks on quantum hardware thus tends to have higher fidelity than running the overall quantum circuit.
Different implementations confirmed this intuition~\cite{Ufrecht:2023,Ying:2023,Singh:2024}. 

In this work, we use the circuit cutting techniques introduced by \citeauthor{Mitarai:2021}~\cite{Mitarai:2021}. 
They showed how to decompose gates of the form $\exp(i \theta A_1\otimes A_2)$ for operators $A_1$ and $A_2$, such that $A_1^2 = I = A_2^2$.
Their decomposition consists of single-qubit and measurements with post-selection~\cite[Lemma~6]{Mitarai:2021}. 

They give an explicit decomposition for the two-qubit CZ gate into single qubit gates and measurements~\cite[Figure~2]{Mitarai:2021}.
We have to run six circuits, which results in ten cases after post-selection of the measurement results. 
This decomposition of the CZ gate gives the tools to decompose any quantum circuit in smaller subcircuits. 

This work tests the feasibility of naive circuit cutting, which means that we will cut individual gates. 
Results from literature state that cutting multiple two-qubit gates simultaneously can be advantageous~\cite{Schmitt:2025}, though these techniques are complex to implement and to some extent obfuscate the impact of individual circuit cuts. 
For the sake of clarity and to keep the models simple, we omit these improvements from our implementations. 
As a result, the results presented in this work can probably be further improved by incorporating these improvements.

\subsection{Probability distribution reconstruction}
\citeauthor{Mitarai:2021} showed that the subcircuits produce the same expectation value as the original circuit.
Following Hoeffding's inequality, approximating the expectation value of a random variable with probability at least $1-\delta$ to within error $\eps$ requires $N=\mathcal{O}(-\log(\delta/2)/\eps^{2})$ samples.

As mentioned in the introduction, approximating the expectation value of a circuit does not always suffice. 
Some quantum algorithms instead require the actual probability distribution or ask for a state with high probability of being found. 
Examples include classifiers with more than two classes and Grover's algorithm with multiple labeled items. 
Oftentimes, the required information cannot be obtained simply from the expectation value. 
Standard circuit cutting techniques might fail to provide usable outcomes in these scenarios. 
These standard circuit cutting techniques can also not directly help in reconstructing the probability distribution of the entire circuit, in part due to the measurements with post-selection. 

In this section, we show that we can reconstruct the original probability distribution by slightly modifying the algorithm by \citeauthor{Mitarai:2021}. 
We claim that we can reconstruct the probability corresponding to \textit{any} of the $2^n$ states to within error $\eps$. 
Following the method, we have at most ten quantum circuits that we have to run. 
We state and prove the theorem for a single circuit cut, though the method directly extends to an arbitrary number of circuit cuts. 
\begin{theorem}
    Let $U$ be an $n$-qubit quantum circuit.
    For every $x\in\{0,1\}^n$ let $p(x) = |\bra{x}U\ket{0}|^2$.
    Let $q_i$ for $i\in\{0,9\}$ be the empirical probability distribution for each of the ten subcircuits, obtained by sampling the subcircuits $N$ times. 
    Let $q(x)$ be the weighted sum of the $q_i$'s for outcome $x\in\{0,1\}^n$. 
    Then, for $N=\mathcal{O}(-\log(\delta/2)/\eps^{2})$, we have that $|p(x) - q(x)|\le \eps$ with probability $1-\delta$ for all $x\in\{0,1\}^n$. 
\end{theorem}
\begin{proof}
    Fix $x\in\{0,1\}^n$. 
    Let $q_i(x)$ be the empirical probability of obtaining $x$ from subcircuit $i$ based on $N$ samples.

    Following Hoeffding's inequality, we know that 
    \begin{equation*}
        \mathrm{Pr}[|q_i(x) - \mathbb{E}(q_i(x))| < \eps] \le 2\exp(-2\eps N).
    \end{equation*}
    We know that $p_i(x) = \sum_{i=0}^{9} \alpha_i \mathbb{E}(q_i(x))$ for $\alpha_i$ depending on the decomposed operation. 
    For the CZ gate, $\alpha_i=\dfrac{1}{2}$ for all $i$. 

    Repeated application of Hoeffding's inequality gives an upper bound on the error of each $q_i(x)$ to within error $\mathcal{O}(\eps)$. 
    By the union bound, we can then approximate $p_i(x)$ to within error $\eps$, proving the theorem. 
\end{proof}
Note that the number of required samples depends on the number of applied cuts. 
Recent work showed that this dependency is submultiplicative~\cite{Schmitt:2025}. 
For our circuit cutting methods we only have to run six subcircuits and then apply post-selection on midcircuit measurements to accommodate for the other four subcircuits. 

Typically, as $n$ grows, $p(x)$ becomes exponentially small, requiring exponentially small $\eps$, and thus exponentially many samples, to reasonably approximate $p(x)$. 
Therefore, reconstructing the entire probability distribution quickly becomes infeasible. 
However, the algorithms we consider produce a probability distribution with a constant number of states of interest with high amplitude (and hence large associated probability).
Our proposed method finds these states with high probability. 
This approach is further motivated by~\cite{Spielman:2009}. 
That work shows that algorithms often work well in practice, even though they have an exponential worst-case running time.

In the remainder of this work, we explore how Hoeffding's inequality and the union bound behave in practice, as both give a lower bound on the success probability. 

\section{Quantum Variational classifier} \label{sec:VC}
This section discusses the quantum variational classifier used to test the feasibility of probability distribution reconstruction in circuit cutting settings. 
We used a locally modified version of the open-source \texttt{TNO Quantum Variational Classifier} package~\cite{TNOVCSOFTWARE:2025}, hereafter referred to as \texttt{VC} and described in \cref{sec:original_VC}.
Next, \cref{sec:DQ_VC} gives the specific adaptations made for this work and the corresponding framework. 
All experiments use the standard Iris dataset, as implemented in the \texttt{tno.quantum.ml.datasets} package~\cite{TNODATASETS:2025}.

\subsection{TNO Quantum Variational Classifier} \label{sec:original_VC}
The original \texttt{VC} package comprises a collection of \texttt{PennyLane}-based modules that implement quantum machine learning (QML) models with different strategies for classical post-processing of measurements. 
\cref{fig:vc_circuits} shows the used quantum circuits, including the data encoding layer and the trainable layers. 
The measured qubits encode the class labels. 
Three different variants of this circuit are used: The \texttt{expected\_value\_model}, the \texttt{modulo\_model}, and the \texttt{parity\_model}.
Their choice may influence the number of classes the classifier can accommodate, the way features are mapped to classes, and the achievable classification accuracy; for further details, see Refs.~\cite{TNOVCSOFTWARE:2025}~and~\cite{TNOVCBLOG:2023}.
\begin{figure}
    \centering
    \includegraphics[width=\linewidth]{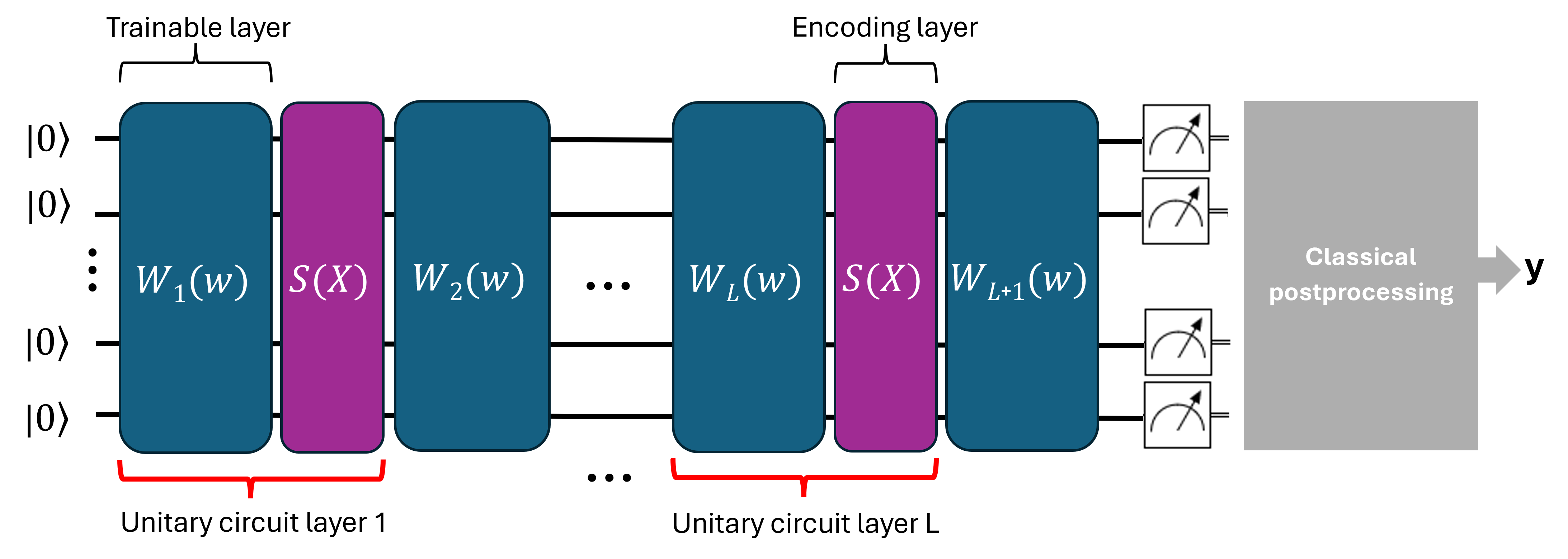}
    \caption{Overview of the QML models implemented in the \texttt{VC} package.
    The models use angle encoding and either expectation values (\texttt{expected\_value\_model}) or probability amplitudes (\texttt{modulo\_model} and \texttt{parity\_model}).}
    \label{fig:vc_circuits}
\end{figure}

The \texttt{expected\_value\_model} uses a mapping from qubit measurements to class labels by computing one Pauli-$Z$ expectation value per class/qubit.
The $n$ qubits of the circuit thus correspond to the number of classes in the dataset.
The measurements used in \cref{fig:vc_circuits} are in those cases replaced by expectation value estimators. 
A softmax function is then applied to the unnormalized expectation values tensor, and the predicted label is assigned to the class with maximum probability. 

The other two models use probability measurements for class assignment. 
The assignment is determined by the decimal representations of measured bitstrings.
In the \texttt{modulo\_model}, each $n$-bit measurement outcome $b$ is assigned a class according to $f(b) = \left[b\right]_{10}\bmod M$, where $M$ is the number of classes, and $[\cdot]_{10}$ is the decimal (base-$10$) representation of the argument.
This model is implemented by partitioning the probability vector into $M$ interleaved streams using strided slicing. 
The probability within each stream is summed to produce unnormalized class scores. 
These aggregated scores are subsequently passed through a softmax to obtain class probabilities.

Finally, \texttt{parity\_model} uses a bitstring-based nonlinear mapping~\cite{Meyer:2023}, where the least significant bits of each outcome index are combined with a parity function over the remaining bits; for the full implementation, see~\cite{Meyer:2023}. 
Note that, this approach enables richer decision boundaries than simple stride-based grouping (as in the \texttt{modulo\_model}) and proves especially useful for parity-based classification tasks or when using custom encodings~\cite{Meyer:2023}.

At the circuit level, the \texttt{VC} package implements quantum unitary transformations composed of alternating trainable mixing layers $W$ and encoding data layers $S$.
The expected values of the measured qubits or measured bitstrings are converted into output features.
\cref{fig:vc_circuit_cut} shows the unitary circuit for $L=1$.
The trainable mixing layers consist of alternating layers of $Z$, $Y$, and $Z$ rotation gates with trainable weights $\mathbf{w}$ (rotation angles) and entangling CNOT gates.
These mixing layers are consistent with \texttt{PennyLane}'s \texttt{StronglyEntanglingLayers}~\cite{Bergholm:2022}.
Dataset features $\mathbf{X}$ are mapped into quantum states using $R_{X}$ angle encoding as in \texttt{PennyLane}'s \texttt{AngleEmbedding}~\cite{Bergholm:2022}.
Note that the flexibility and expressibility of the unitary circuit can be conveniently tuned by performing $L$ repetitions of the $W_{i}(\mathbf{w}) \otimes S(\mathbf{X})$ block structure, each with their own trainable parameters~\cite{Hubregtsen:2021}.
\begin{figure}
    \centering
    \includegraphics[width=1\linewidth]{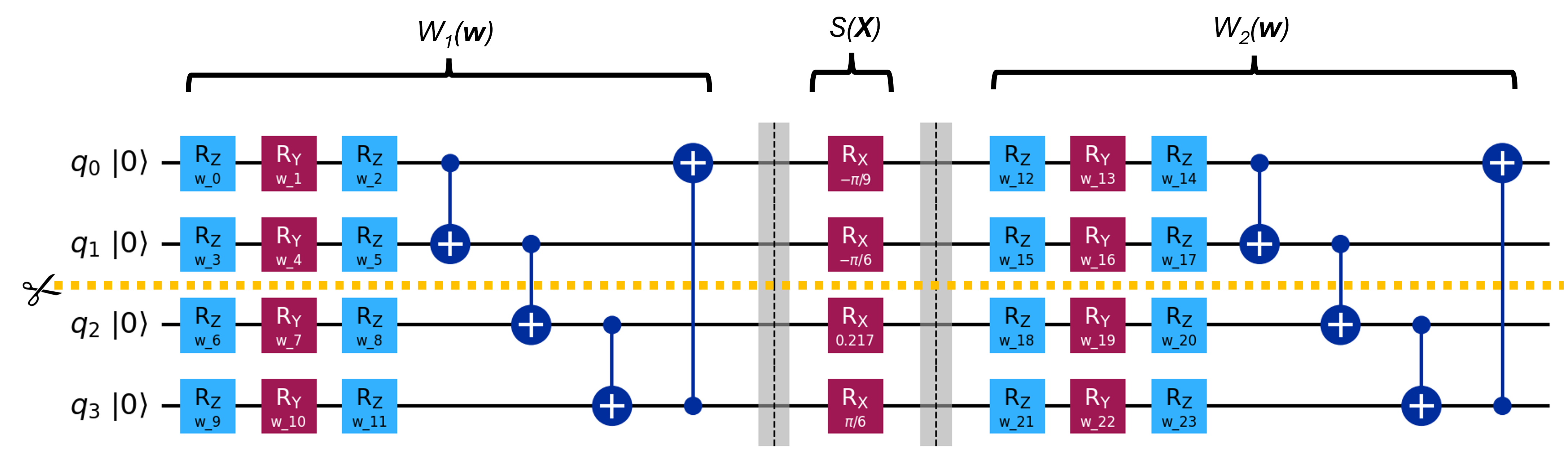}
    \caption{Structure of the used variational quantum circuits used in this work.
    Shown is the detailed implementation of one unitary circuit layer of \cref{fig:vc_circuits} for $L=1$.
    Trainable layers are implemented using \texttt{PennyLane}'s \texttt{StronglyEntanglingLayers}~\cite{Bergholm:2022} with a total of 24 trainable parameters ($\mathbf{w}$), while classical input features ($\mathbf{X}$) are encoded via \texttt{PennyLane}'s \texttt{AngleEmbedding}~\cite{Bergholm:2022}. 
    The illustrated circuit encodes the preprocessed feature vector $\mathbf{X}\!=\![5.9, 3.0, 4.2, 1.5]$ from the Iris dataset (class label $y\!=\!\rm{versicolor}$).
    The corresponding circuit cut strategy employed throughout is also depicted, showing 4 cut CNOT gates.}
    \label{fig:vc_circuit_cut}
\end{figure}

\subsection{\emph{DivideQuantum} Variational Classifier}\label{sec:DQ_VC}
This section describes the developed framework and the changes made to the original \texttt{PennyLane}-based \texttt{VC} package. 
We wish to use the functionalities offered by the \texttt{qiskit-addon-cutting} package~\cite{QISKITADDONCUTTING:2024}.
The following changes are made to be able to use these functionalities:
\begin{enumerate}
    \item{Implementation of dedicated classes that inherit functionalities from the original \texttt{VariationalClassifier}, \texttt{ProbabilityModel}, and \texttt{ExpectedValueModel}, adapted to accept circuit cut options and ensure full interoperability with \texttt{Qiskit Primitives}~\cite{QISKIT:2024} and the \texttt{qiskit-addon-cutting} package.}
    
    \item{Implementation of a \texttt{PennyLane}-to-\texttt{Qiskit} converter module that generates symbolic \texttt{Qiskit QuantumCircuit} objects with circuit \texttt{Parameter} instances, in a manner similar to the \texttt{PennyLane-Qiskit} plugin~\cite{Bergholm:2022}. Special attention is given to the conversion between big-endian qubit ordering in \texttt{PennyLane} and little-endian ordering in \texttt{Qiskit}.}
    
    \item{Implementation of a custom \texttt{torch.autograd.Function} to enable seamless integration between \texttt{Qiskit} and \texttt{PyTorch}~\cite{Jason:2024}, analogous to \texttt{TorchConnector} in \texttt{qiskit-machine-learning}~\cite{Sahin:2025}. Our implementation relies on the Parameter Shift Rule~\cite{Schuld:2019} to compute exact gradients during the backward pass.}
\end{enumerate}

\begin{figure}
    \centering
    \includegraphics[width=0.9\linewidth]{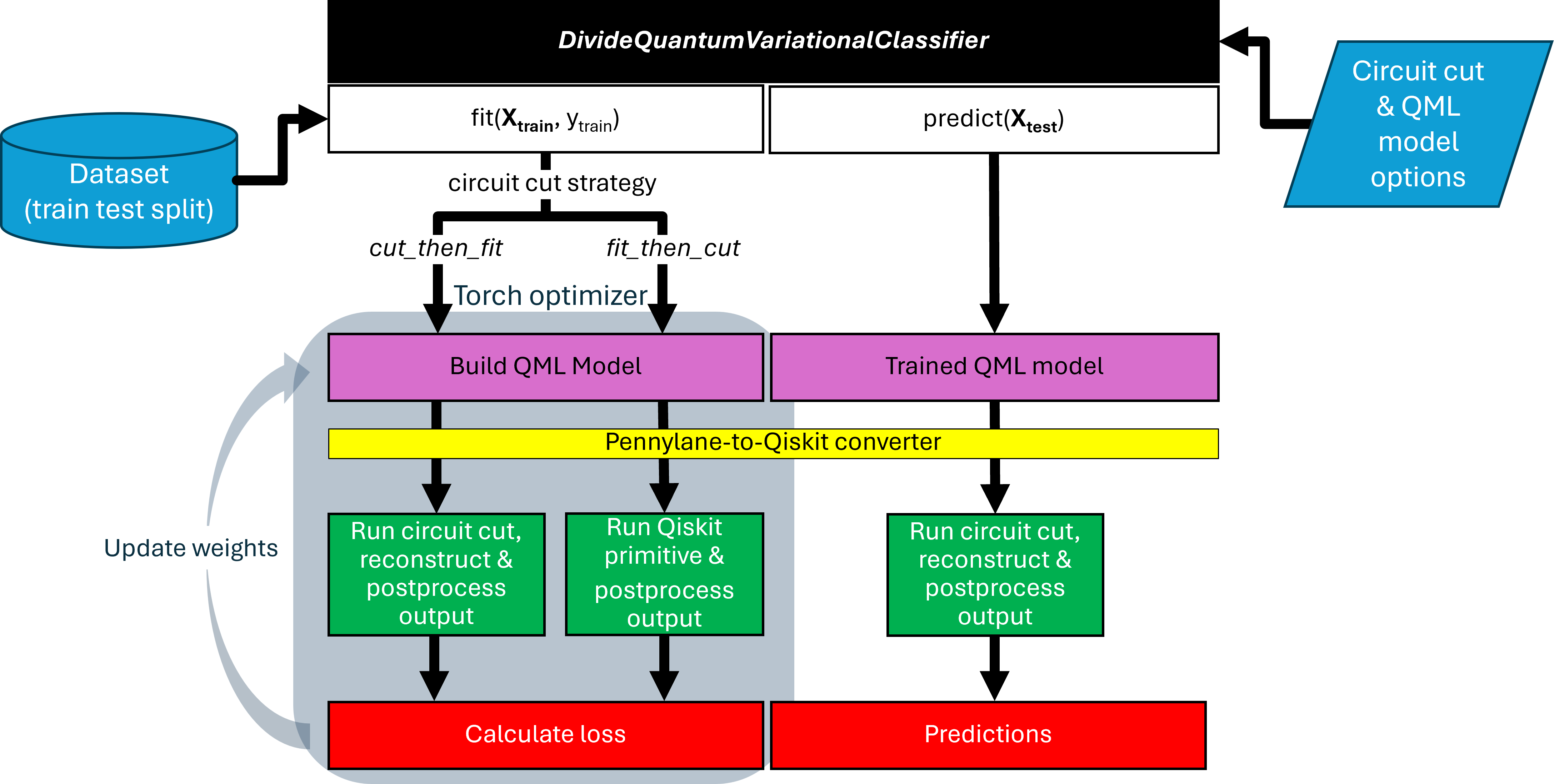}
    \caption{Main structure and design of the \texttt{DivideQuantumVariationalClassifier} implemented in this work. It extends the core functionality of the original \texttt{VC} package~\cite{TNOVCSOFTWARE:2025} to support circuit gate cuts and integrates seamlessly with the \texttt{qiskit-addon-cutting} package~\cite{QISKITADDONCUTTING:2024}. It features a dedicated \texttt{Pennylane}-to-\texttt{Qiskit} symbolic circuit converter developed in-house and is fully compatible with \texttt{PyTorch} tensors, optimizers, and loss functions via a specialized \texttt{torch.autograd.Function} connector.}
    \label{fig:dq_variational_classifier}
\end{figure}

\cref{fig:dq_variational_classifier} outlines the main structure and design of the final \texttt{DivideQuantumVariationalClassifier} implemented within our modified \texttt{VC} package, hereafter referred to as \texttt{DivideQuantum VC}, \texttt{DQ VC} for short. 
Similarly to the original \texttt{VariationalClassifier}, its implementation follows the scikit-learn estimator API~\cite{SCIKITLEARN:2025}, ensuring full compatibility with standard machine learning workflows via \texttt{fit} and \texttt{predict} methods.
Note the choice for two different \texttt{fit} methodologies. 
These two choices relate to the different circuit cut strategies combined with the training stage of the algorithm: 
\begin{itemize}
    \item{\texttt{fit\_then\_cut}: When this option is selected, QML models are trained as usual, without circuit cuts, using either \texttt{Qiskit Primitives} or the original \texttt{PennyLane QNode} \texttt{VC} workflow for circuit evaluations.
    The effects of gate cuts are only introduced during the validation phase, where the trained QML model is evaluated using \texttt{predict} with circuit cuts applied.}

    \item{\texttt{cut\_then\_fit}: With this option, circuit cuts and associated subexperiments are executed at every QML training iteration.
    Circuit outputs and losses are computed using the reconstructed expectation values or probability distributions, depending on the selected QML model. 
    Circuit cuts are also applied during validation, where the trained model is evaluated on the test set using \texttt{predict}.
    The \texttt{cut\_then\_fit} option is only compatible with the newly implemented \texttt{Qiskit} workflow. 
    While it fully incorporates the effects of circuit cuts, it is computationally very intensive. 
    For example, in a typical training iteration or backward pass using the simplest QML model shown in \cref{fig:vc_circuit_cut}, computing the full circuit gradient requires at most 62208 circuit evaluations.
    This corresponds to (2 angle shifts of $\pm~\pi/2$ per parameter) $\times$ (24 parameters) $\times$ ($6^4$ subcircuit executions involving 4 CNOT gate cuts). 
    In practice, we observed that after transpilation using \texttt{Qiskit}'s \texttt{generate\_preset\_pass\_manager} some variational parameters contributed negligibly to the gradient ($<1 \times 10^{-9}$). 
    To reduce the number of unnecessary computations, we therefore implemented a masking scheme within \texttt{torch.autograd.Function}. 
    This assigns Boolean masks (\texttt{True} or \texttt{False}) to each parameter to determine whether its gradient should be evaluated. 
    To account for potential changes in circuit structure or parameter relevance during optimization, all masks are reset every 10 optimization iterations, allowing previously masked parameters to re-enter the gradient computation process if their contribution becomes significant ($>1 \times 10^{-9}$).}
\end{itemize}

\cref{sec:Results} presents the results of the classification experiments for both strategies and the three different variants of the model. 
Recall that \cref{fig:vc_circuit_cut} shows the circuit used in the experiments for $L=1$ and that this model corresponds to the simplest model available within the \texttt{DivideQuantum VC} package with the lowest number of gate cuts (4 CNOT gates).

\section{Results} \label{sec:Results}
This section provides the results of the conducted circuit cutting experiments and the backends used. 
The remaining sections each provide the results for one of the circuit cutting experiments conducted. 
We conclude with the results of the hardware implementation and the comparison with the noisy simulation. 
\cref{tab:model_scores} gives a performance summary for all conducted experiments. 

\subsection{Experimental setup}
We conducted various experiments to test our circuit cutting implementations and to see how well our probability reconstruction approach works in practice. 

For all experiments, we used the standard Iris dataset~\cite{TNODATASETS:2025}. 
In the preprocessing phase, a train-test split with a 25\% test size was applied to the dataset, resulting in 112 training samples and 38 test samples. 
All models were trained using the \texttt{Pytorch Adagrad} optimizer with a learning rate of 0.1 and a weight decay of $1 \times 10^{-4}$.

Most results are obtained using quantum simulators, specifically the \texttt{Qiskit AerSimulator} and \texttt{PennyLane default.qubit} device, with a shot count of $2^{12}$. 
Additionally, for the \texttt{parity\_model} we also ran noisy simulations and implemented the circuits on quantum hardware. 
For these results, we used the 127 qubit \texttt{ibm\_strasbourg} quantum computer. 
The noisy simulation used a noisy model based on this same device. 
Note that when interfacing with \texttt{qiskit-addon-cutting}, the \texttt{num\_samples} option, i.e., the number of samples to draw from the quasi-probability distribution, was set to 2000; see~\cite{QISKITADDONCUTTING:2024} for more details.

The first experiment validates the correctness of circuit cutting and the developed framework by considering two different quantum circuits. 
All remaining experiments consider the quantum variational classifier outlined in the previous section. 
Specifically, we consider the three different models (\texttt{expected\_value\_model}, \texttt{modulo\_model}, and \texttt{parity\_model}) and for each of these three, we employ either the \texttt{fit\_then\_cut} strategy or the \texttt{cut\_then\_fit} strategy and compare the results with an uncut circuit. 
As quantum hardware typically is restricted in size, the latter strategy is of particular interest. 

For validation, we use two simple test cases, shown in \cref{fig:validation-circuits}. 
For the remaining experiments, we use the four-qubit quantum circuits involved in the evaluation and training of a Variational Classifier (VC), see \cref{sec:VC}. 
We restricted ourselves to the four-qubit case due to the computational intensity of simulating the large number of circuits that is required when using circuit cutting.
We give a detailed presentation of the result for the VC trained using the \texttt{expected\_value\_model} and then use a briefer format for the results for the VCs trained using the other two models, as they show a similar performance. 
\begin{figure}
  \centering
  \begin{subfigure}[t]{0.4\linewidth}
    \centering
      \begin{quantikz}
        \lstick{$\ket{0}$} & \gate{H} & \ctrl{1} & \ghost{U_1}\qw \\
        \lstick{$\ket{0}$} & \ghost{U_1}\qw & \targ{}  &\qw
      \end{quantikz}
      \caption{The two-qubit GHZ circuit.}
    \label{fig:GHZ circuit}
  \end{subfigure}
  \begin{subfigure}[t]{0.4\linewidth}
    \centering
      \begin{quantikz}
        \lstick{$\ket{0}$} & \gate{U_1} & \gate[2]{U_3} &\gate{U_4} &\qw \\
        \lstick{$\ket{0}$} & \gate{U_2} &               &\gate{U_5} &\qw
      \end{quantikz}
      \caption{A randomly generated circuit.}
    \label{fig:random-circuit}
  \end{subfigure}
  \caption{The circuits used for validating the circuit cutting implementations.
  For the random circuit, each gate $U_i$ represents a Haar-random unitary of appropriate size.}
  \label{fig:validation-circuits}
\end{figure}
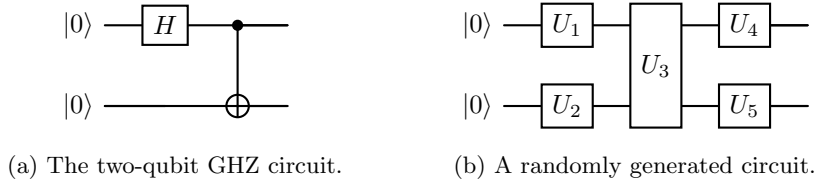

These experiments all use perfect simulations.
We continue with a noisy simulation of an uncut circuit and a cut circuit, both using the \texttt{parity\_model}.
We then conclude with a hardware implementation of both an uncut and a cut circuit, again using the \texttt{parity\_model}.  

\subsection{Validation}\label{sec:Validation}
We validated the performance of the circuit cutting methods using the two circuits shown in \cref{fig:validation-circuits}.
We compared the expectation value for different observables and computed the absolute error for the cut circuit with respect to the uncut circuit. 
Every run uses $2^{12}$ measurement rounds and we average the results over $100$ runs.

\cref{fig:validation-expectation-results} shows the absolute errors for the observables for both circuits. 
We see a small absolute error for each of the observables, which most likely originate from the statistical sampling nature of quantum algorithms and the low number of circuit evaluations. 
These effects get magnified by the subcircuit sampling introduced by the circuit cutting methods. 
The random circuit shows smaller errors than the circuit producing the GHZ state, which might be due to the broader distribution of the circuit. 
\begin{figure}
  \centering
  \begin{subfigure}[t]{0.49\linewidth}
    \centering
    \includegraphics[width=\linewidth]{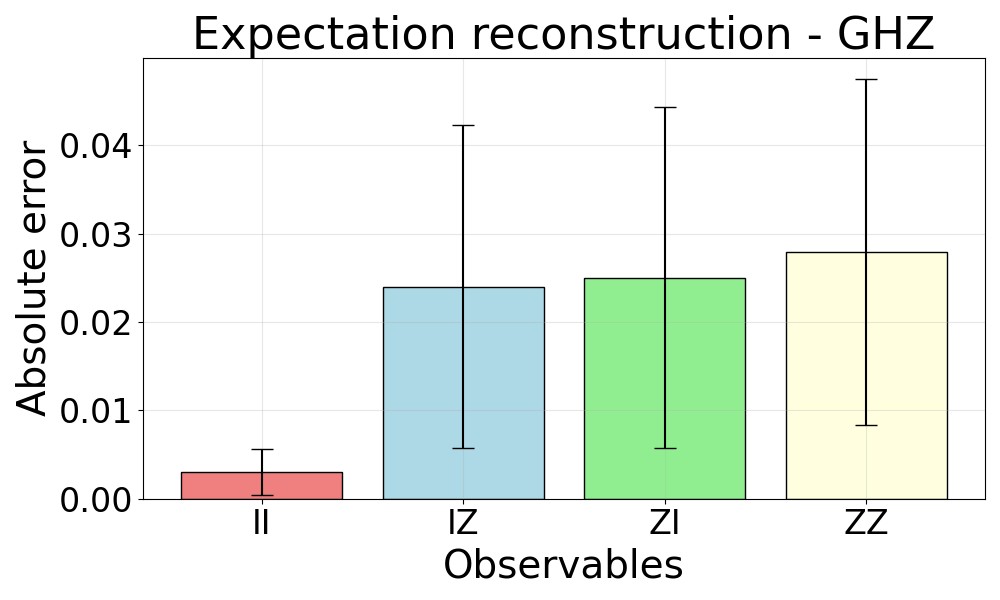}
    \caption{Results for the two-qubit GHZ circuit.}
    \label{fig:expectation-reconstruction-GHZ}
  \end{subfigure}
  \begin{subfigure}[t]{0.49\linewidth}
    \centering
    \includegraphics[width=\linewidth]{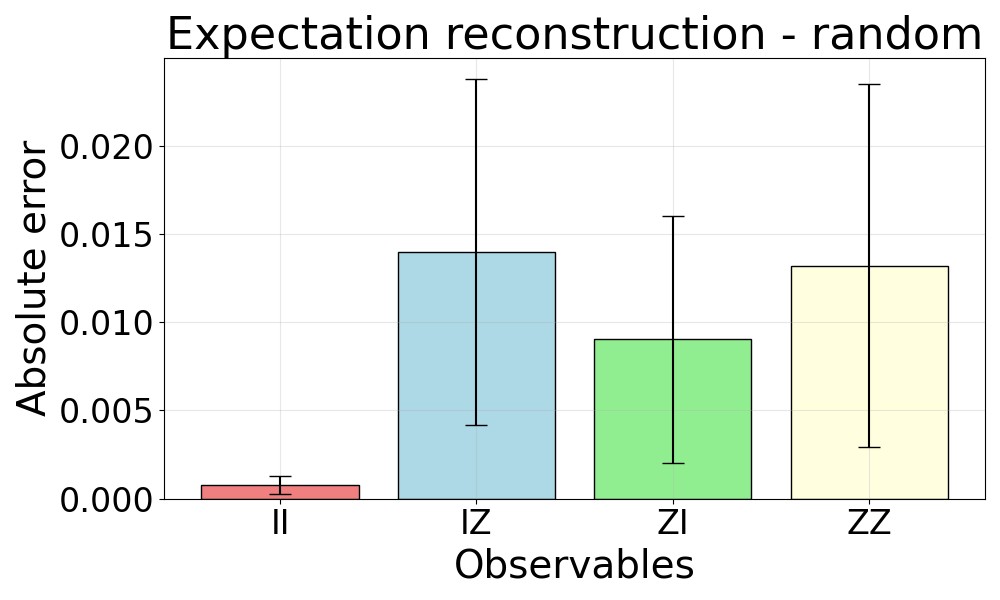}
    \caption{Results for the two-qubit random circuit.}
    \label{fig:expectation-reconstruction-random}
  \end{subfigure}
  \caption{The average and standard deviation in absolute error when using circuit cutting to reconstruct the expectation values of the depicted observables, and 100 runs.}
  \label{fig:validation-expectation-results}
\end{figure}

\cref{fig:validation-probability-results} shows the results the total reconstructed probability distribution averaged over $100$ runs for both circuits. 
The average total deviation for the GHZ circuit was $0.2\%$ and for the random circuit this was $0.8\%$, with standard deviations of $0.2\%$ and $0.3\%$ respectively. It is interesting to note that for these experiments the probability distribution reconstruction seems to be more robust than the expectation value reconstruction.
\begin{figure}
  \centering
  \begin{subfigure}[t]{0.49\linewidth}
    \centering
    \includegraphics[width=\linewidth]{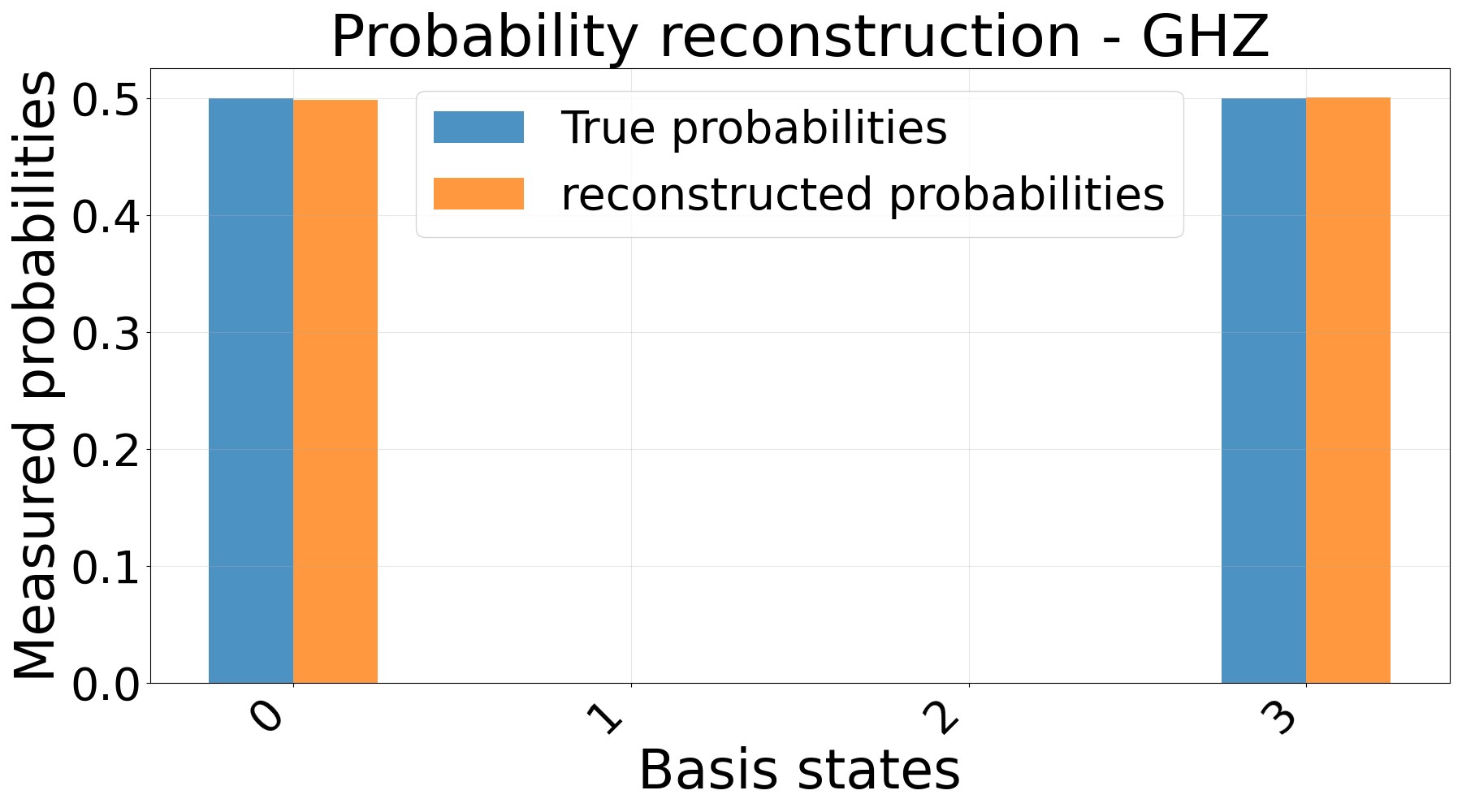}
    \caption{Results for the two-qubit GHZ circuit.}
    \label{fig:probability-reconstruction-GHZ}
  \end{subfigure}
  \begin{subfigure}[t]{0.49\linewidth}
    \centering
    \includegraphics[width=\linewidth]{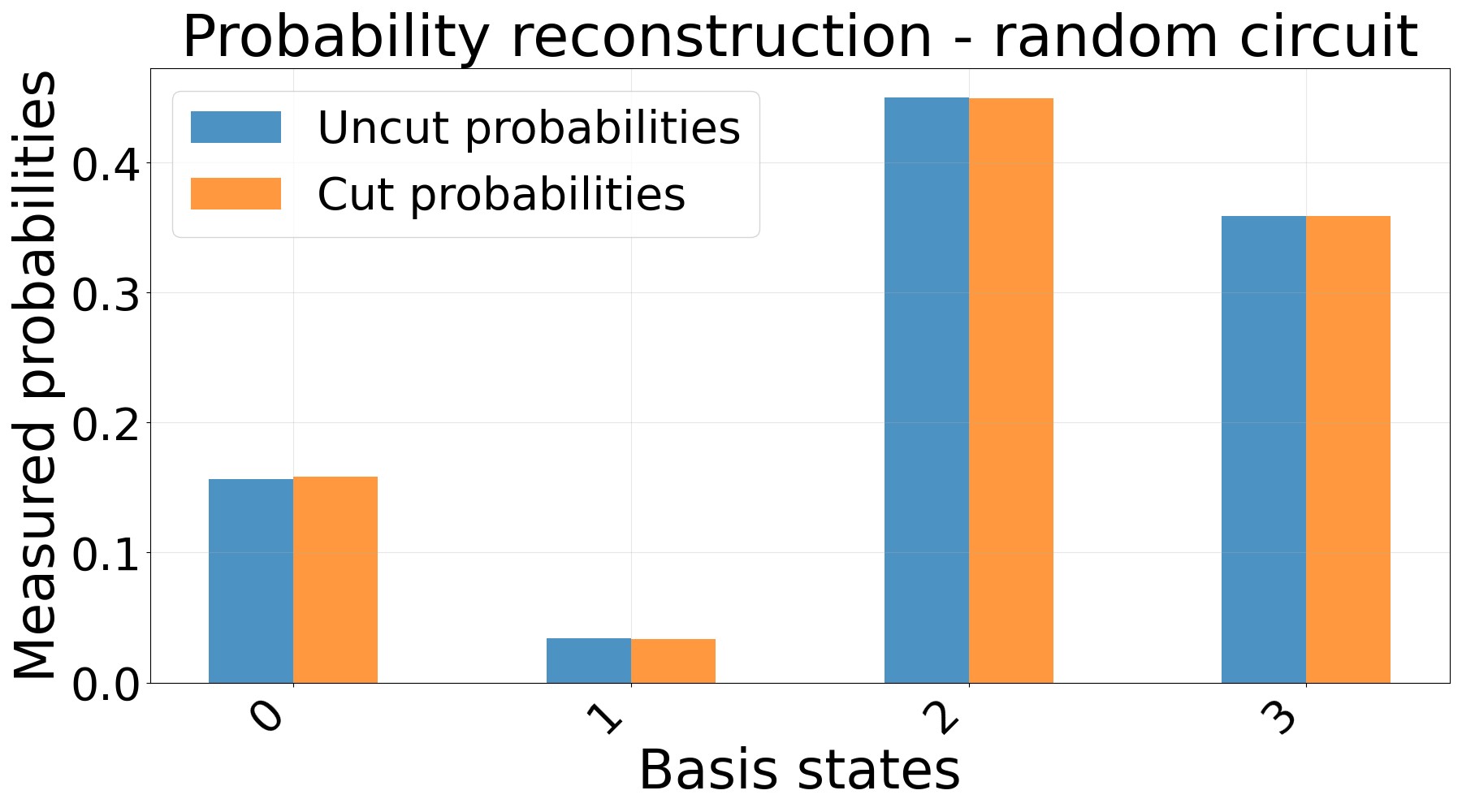}
    \caption{Results for the two-qubit random circuit.}
    \label{fig:probability-reconstruction-random}
  \end{subfigure}
  \caption{The deviation between the true probability distribution and the empirical probability distribution obtained when using circuit cutting.}
  \label{fig:validation-probability-results}
\end{figure} 

The next sections test the performance of the different models used for circuit cutting the variational classifier. 

\subsection{VC with the \texttt{expected\_value\_model}}\label{sec:expected_value}
This section shows the results obtained for the variational classifier that uses expectation values to assign scores to different classes.
For the first results, we have fitted the classifier without using circuit cutting and then tested the performance of the cutting implementations using the validation data sets.
For the second results, we considered a more practical setting, where also the training was performed using cut circuits. 
Due to the high simulation load, we have used a limited training phase for the \texttt{cut\_then\_fit} strategy. 
Still, the results show great alignment with the uncut results. 

\subsubsection{Fit then cut}\label{sec:expected_value_fit_and_cut}
Here we show the performance of \texttt{fit\_then\_cut} strategy (\cref{sec:DQ_VC}) using circuits generated by the \texttt{expected\_value\_model} classifier. 
For this we used a validation data set of 38 elements.
\cref{fig:expectation-reconstruction-VC} shows the expectation value produced by these circuits, and the observed deviation between the values produced with and without cutting.
The three observables show the measured expectation values for the different classes, the accumulated error shows the error for a single randomly chosen instance, and the average error shows the average accumulated error over all 38 elements.
\begin{figure}
    \centering
    \includegraphics[width=0.5\linewidth]{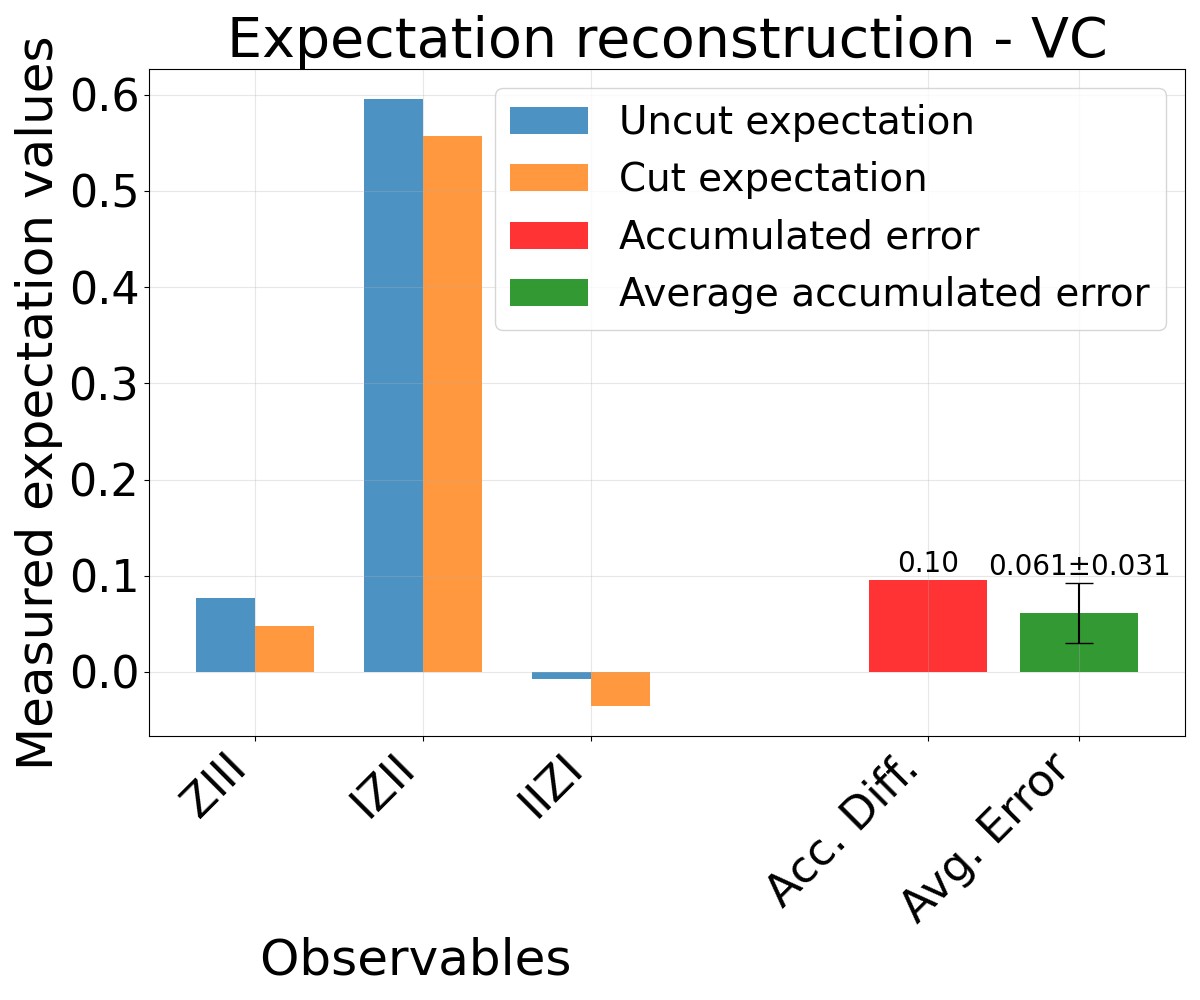}
    \caption{The deviation in expectation values for circuits generated by the \texttt{expected\_value\_model} classifier.
    On the left we see a comparison between the produced expectation values with and without circuit cutting for a single circuit. The accumulated error shown is the total absolute deviation for this single circuit. The average accumulated error shown is the aggregate of similar results for the whole test set of 38 circuits.}
    \label{fig:expectation-reconstruction-VC}
\end{figure}

For the \texttt{expected\_value\_model} classifier, the expectation values of the depicted observables directly translate to the label assigned to a sample.
\cref{fig:class-reconstruction-EVVC_all} shows the aggregated confusion matrix of the predicted and true labels of $100$ different train-test sets. 
%
\begin{figure}
    \centering
    \includegraphics[width=\linewidth]{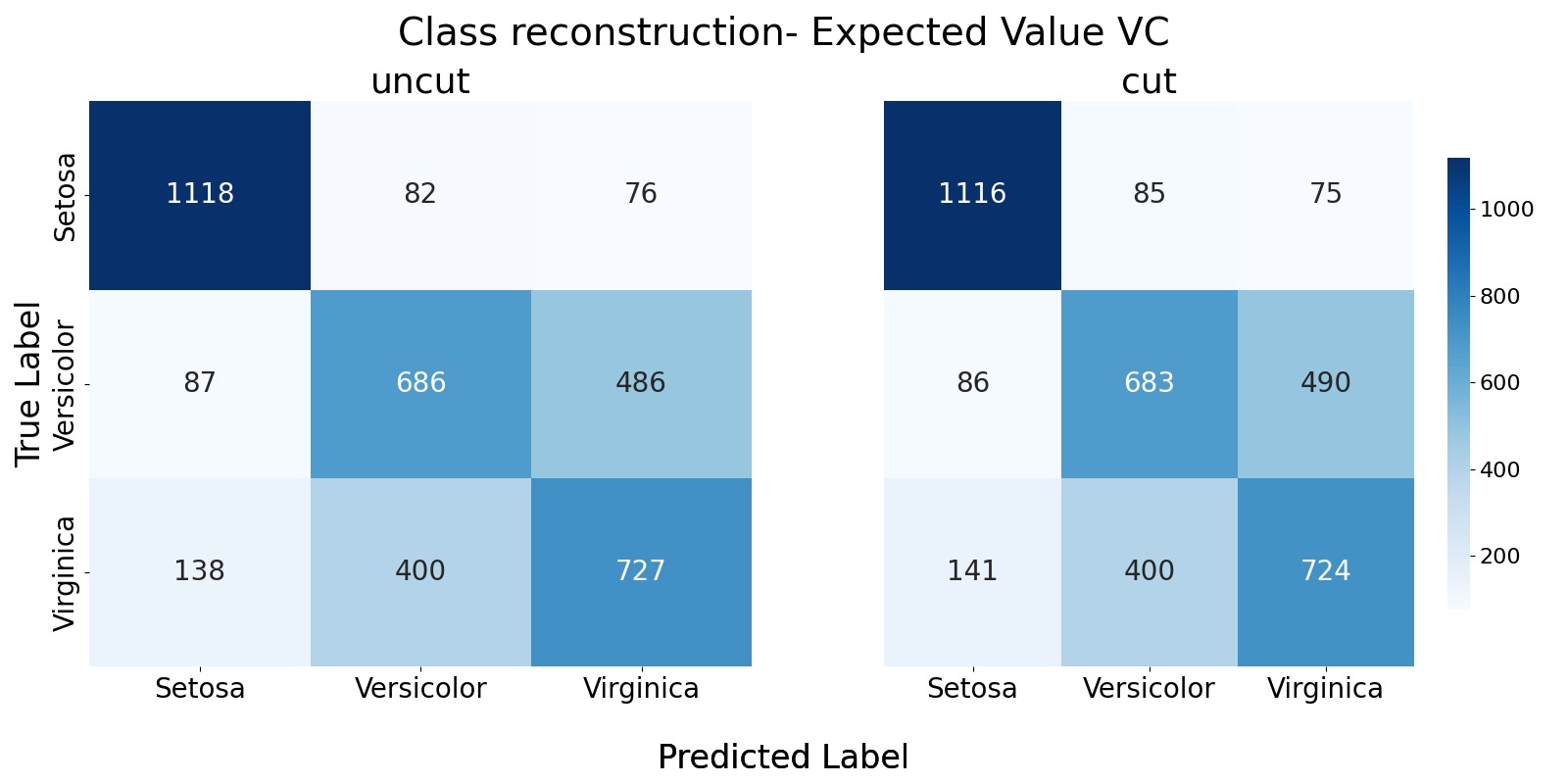}
    \caption{Confusion matrices comparing the classifications produced by the cut and uncut version, using 100 validation sets of 38 elements}
    \label{fig:class-reconstruction-EVVC_all}
\end{figure}
We see that performance of the classifier is barely affected by the usage of circuit cutting.
The \emph{Setosa} samples are often classified correctly, whereas for the \emph{Versicolor} and \emph{Virginica} samples, errors are made. 
Still, the results between the cut and uncut versions align, again showing the correctness of the circuit cutting methods. 

\subsubsection{Cut then fit} \label{sec:expected_value_cut_and_fit}
We now consider the results for the same \texttt{expected\_value\_model} classifier, but now using the \texttt{cut\_then\_fit} strategy.
Note that, due to its high computational cost, the results presented in this section were obtained by training the QML models for only one epoch (5 iterations), starting from previously optimized model parameters derived from standard \texttt{PennyLane-based} \texttt{VC} runs.
This initial phase consisted of training the uncut models for 10 epochs with a batch size of 25 for a total of 50 iterations. 
The final parameters were then used to initialize the \texttt{cut\_then\_fit} experiments.
Note further that, unlike \texttt{fit\_then\_cut}, the additional computational overhead of \texttt{cut\_then\_fit} also prevented us from performing multiple trainings with different random seeds. Therefore, all experiments were conducted using a fixed random seed of 42 to ensure reproducibility. 
\cref{fig:cut_and_fit_expected_value} shows the results for the \texttt{expected\_value\_model} classifier in terms of the confusion matrix.
We also included the uncut model trained for 11 epochs, assuring the final results are obtained using the same number of training iterations. 
Interestingly, the cut version seems outperform the uncut version for this specific instance. 
We expect that results from statistical inaccuracies. 

The confusion matrix shows that the training loss history of the \texttt{cut\_then\_fit} \texttt{DQ VC} model is largely consistent with that of the original uncut models. 
The classification scores are also comparable: the circuit-cut \texttt{DQ VC} model achieved 78.9\% (\cref{tab:model_scores}), while the uncut models reached 73.7\% (after 11 epochs) and 81.6\% (after 10 epochs), respectively. 
Overall, as with \texttt{fit\_then\_cut}, the approximations introduced by the \texttt{fit\_then\_cut} strategy and their impact on the classification results are minimal.

A possible explanation can be found in the \textit{warm start} using already trained parameters. 
Still, we continued the training for five more iterations. 
The loss for \texttt{fit\_then\_cut} follows roughly the same pattern as the loss for the uncut variational classifier. 
\begin{figure}
  \centering
  \begin{subfigure}[b]{\linewidth}
    \centering
    \includegraphics[width=0.99\linewidth]{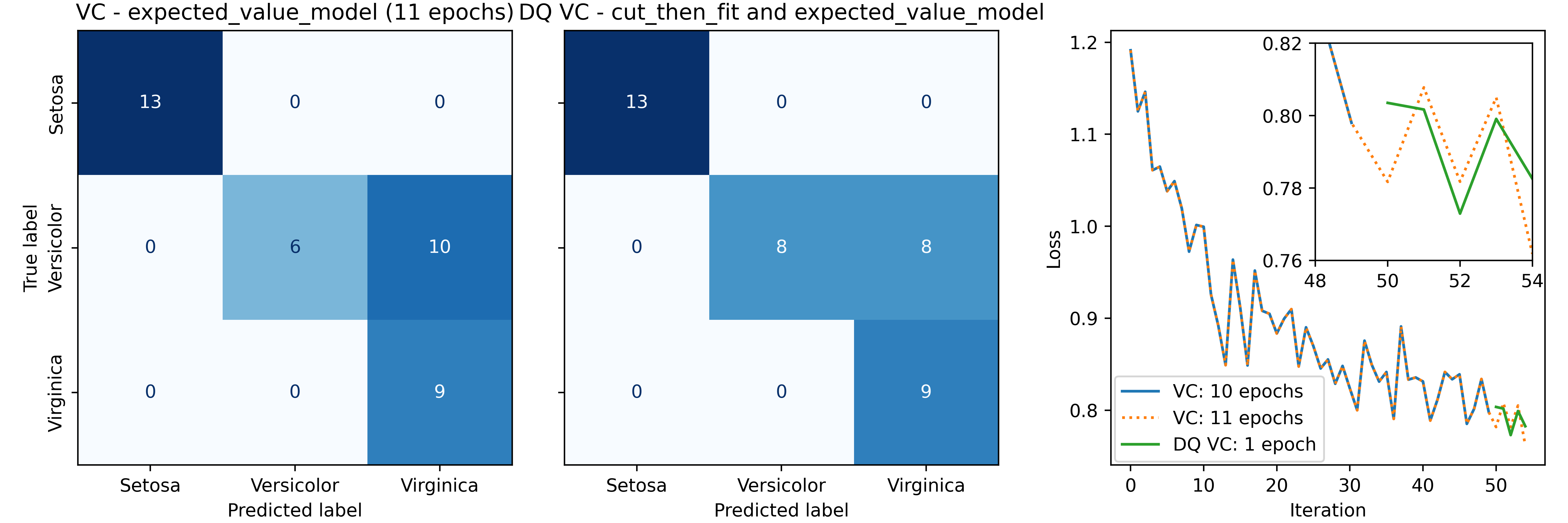}
    \caption{}
    \label{fig:cut_and_fit_expected_value}
  \end{subfigure}
  \vskip\baselineskip
  \begin{subfigure}[b]{\linewidth}
    \centering
    \includegraphics[width=0.99\linewidth]{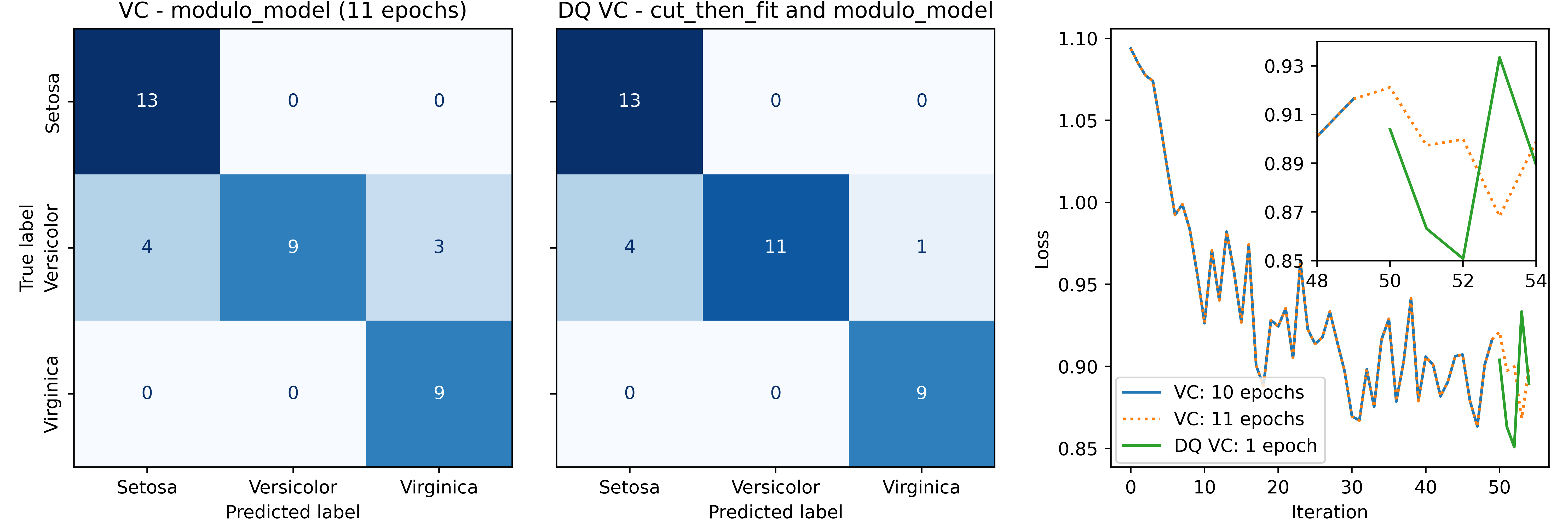}
    \caption{}
    \label{fig:cut_and_fit_modulo_model}
  \end{subfigure}
  \vskip\baselineskip
  \begin{subfigure}[b]{\linewidth}
    \centering
    \includegraphics[width=0.99\linewidth]{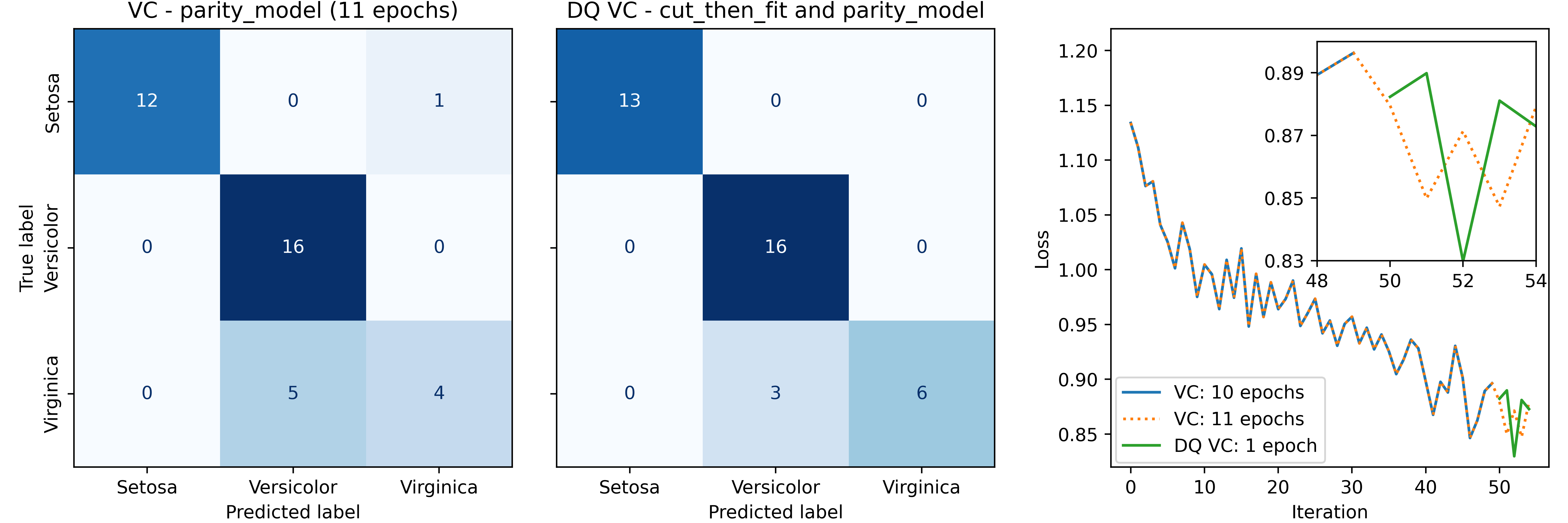}
    \caption{}
    \label{fig:cut_and_fit_parity_model}
  \end{subfigure}
  \caption{Confusion matrices and training loss histories of the various QML models trained with the \texttt{cut\_then\_fit} strategy:
  (a) \texttt{expected\_value\_model}; (b) \texttt{modulo\_model}; (c) \texttt{parity\_model}. 
  \texttt{VC} refers to the original \texttt{PennyLane}-based \texttt{VC} runs without circuit cuts, while \texttt{DQ VC} denotes the results obtained using the modified \texttt{VC} package developed in this work with circuit cutting. 
  \texttt{DQ VC} models are trained for 5 iterations only (1 epoch and batch size 25) using as starting parameters those previously optimized with the original uncut \texttt{VC} for 10 epochs. 
  For ease of comparison, the iteration indices (x-axis) of the \texttt{DQ VC} models are shifted by 50 in the loss history plots.}
  \label{fig:cut_and_fit_performance}
\end{figure}

\subsection{VC with the \texttt{modulo\_model}}
Here we show the results obtained for the classifier using the \texttt{modulo\_model}. 
Again, for our results, we have fitted the classifier using the \texttt{fit\_then\_cut} and the \texttt{cut\_then\_fit} strategy.

\subsubsection{Fit then cut} \label{sec:modulo_fit_and_cut}
Here we show the performance of the \texttt{fit\_then\_cut} implementation using circuits generated by the \texttt{modulo\_model} classifier. 
\cref{fig:probability-reconstruction-VC-modulo} summarizes the result for the reconstructed probability distribution with and without circuit cutting for the circuit shown in \cref{fig:vc_circuit_cut}.
First, for a single instance, the reconstructed probability distribution is compared with the original one, and the accumulated error over all basis states is shown. Additionally, we show the accumulated error averaged over all 38 instances. 
We see that despite the small deviations, the general form of the probability distribution is still obtained, the main goal of this circuit cutting technique. 
\begin{figure}
    \centering
    \begin{subfigure}[b]{\linewidth}
        \centering
        \includegraphics[width=\linewidth]{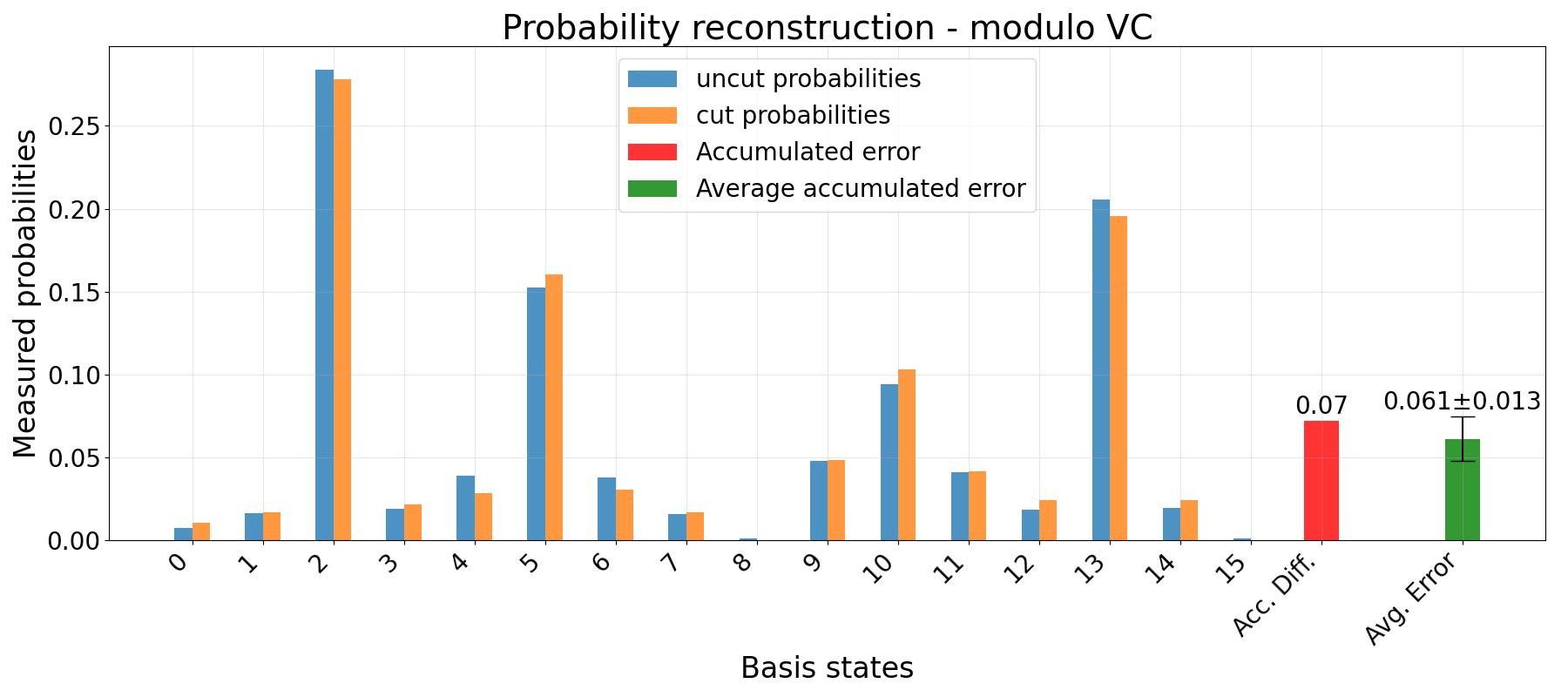}
        \caption{Probability reconstruction for the modulo classifier.}
        \label{fig:probability-reconstruction-VC-modulo}
    \end{subfigure}
    \begin{subfigure}[b]{\linewidth}
        \centering
        \includegraphics[width=\linewidth]{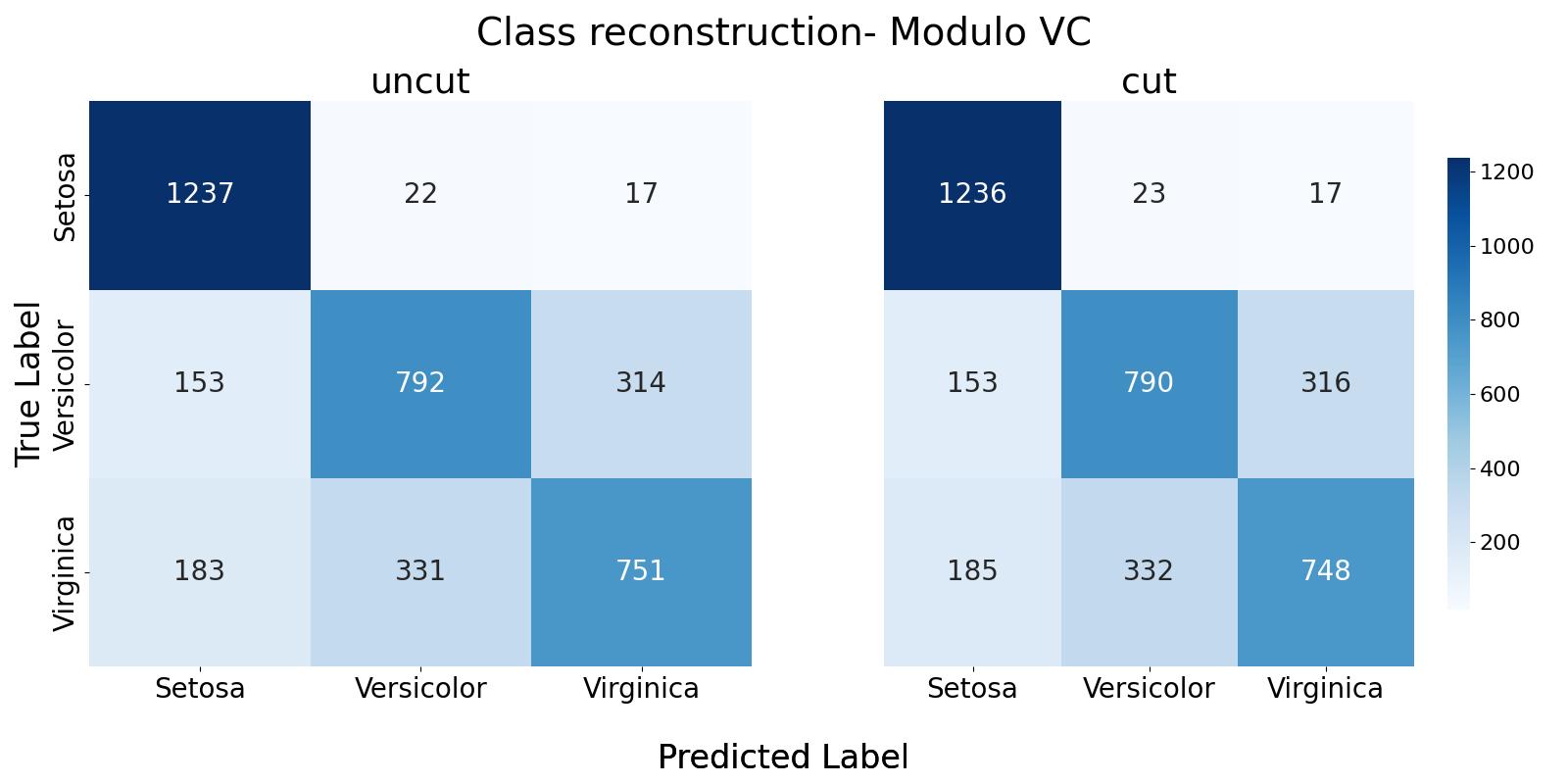}
        \caption{Aggregated results for 100 validations sets of 38 elements}
    \label{fig:class-reconstruction-MVC-all}
    \end{subfigure}
    \caption{The results for the \texttt{modulo\_model}.
    The deviation in the reconstructed probability distribution is given, as well as the confusion matrix for 100 different train-test instances. 
    For both figures, the uncut and cut versions are compared, where the cut version uses the \texttt{modulo\_model} and the \texttt{fit\_then\_cut} strategy. }
    \label{fig:probability-reconstruction}
\end{figure}

It is interesting to see that the average accumulated error for the probability distribution reconstruction is similar to the average accumulated error for the expectation value, as you would suspect in that model that statistical variances will play a smaller role. 
Additionally, the accumulated error is now taken over 16 basis states, rather than over 4 expectation values. 
\cref{fig:class-reconstruction-MVC-all} shows the aggregated confusion matrix for 100 different train-test sets.
We again observe a strong overlap between the uncut and cut versions. 

\subsubsection{Cut then fit} \label{sec:modulo_cut_and_fit}
\cref{fig:cut_and_fit_modulo_model} shows the results of the QML fits done with the \texttt{modulo\_model} classifier in combination with the \texttt{cut\_then\_fit} strategy. 
Note that the experiments were conducted in a similar manner as for the \texttt{expectation\_value\_model}. 
\cref{tab:model_scores} summarizes the final predicted classification scores. 

Despite the more noticeable differences in the training loss histories of the \texttt{DQ VC} and uncut \texttt{VC} models, the predicted classification scores remain quite consistent -- similar to results obtained with \texttt{fit\_then\_cut}.
Specifically, the cut \texttt{DQ VC} model achieves a score of 86.8\%, whereas the corresponding value for the uncut model is 81.6\% after 11 epochs.
Excellent agreement is also found when comparing the corresponding confusion matrices, further attesting to the robustness of our circuit cutting strategy. 
Indeed, \cref{fig:cut_and_fit_modulo_model} reveals that the \texttt{modulo\_model} performs significantly better in classifying \emph{Versicolor} instances compared to \texttt{expected\_value\_model}, cf. \cref{fig:cut_and_fit_expected_value}.

\subsection{VC with the \texttt{parity\_model}}
Here we show the results obtained for the classifier using the \texttt{parity\_model}. 
Again, for our results, we have trained the classifier using the \texttt{fit\_then\_cut} and the \texttt{cut\_then\_fit} strategy.

\subsubsection{Fit then cut} \label{sec:parity_fit_and_cut}
Here we show the performance of the \texttt{fit\_then\_cut} implementation using circuits generated by the \texttt{parity\_model} classifier. 
\cref{fig:probability-reconstruction-VC-parity} summarizes the result for the reconstructed probability distribution with and without circuit cutting for the circuit shown in \cref{fig:vc_circuit_cut}.
The shown probability distribution and accumulated error again is for a single instance, whereas the average accumulated error is taken over all 38 instances. 
Again, the general form of the probability distribution is correctly found. 
\begin{figure}
    \centering
    \begin{subfigure}[b]{1\linewidth}
        \centering
        \includegraphics[width=\linewidth]{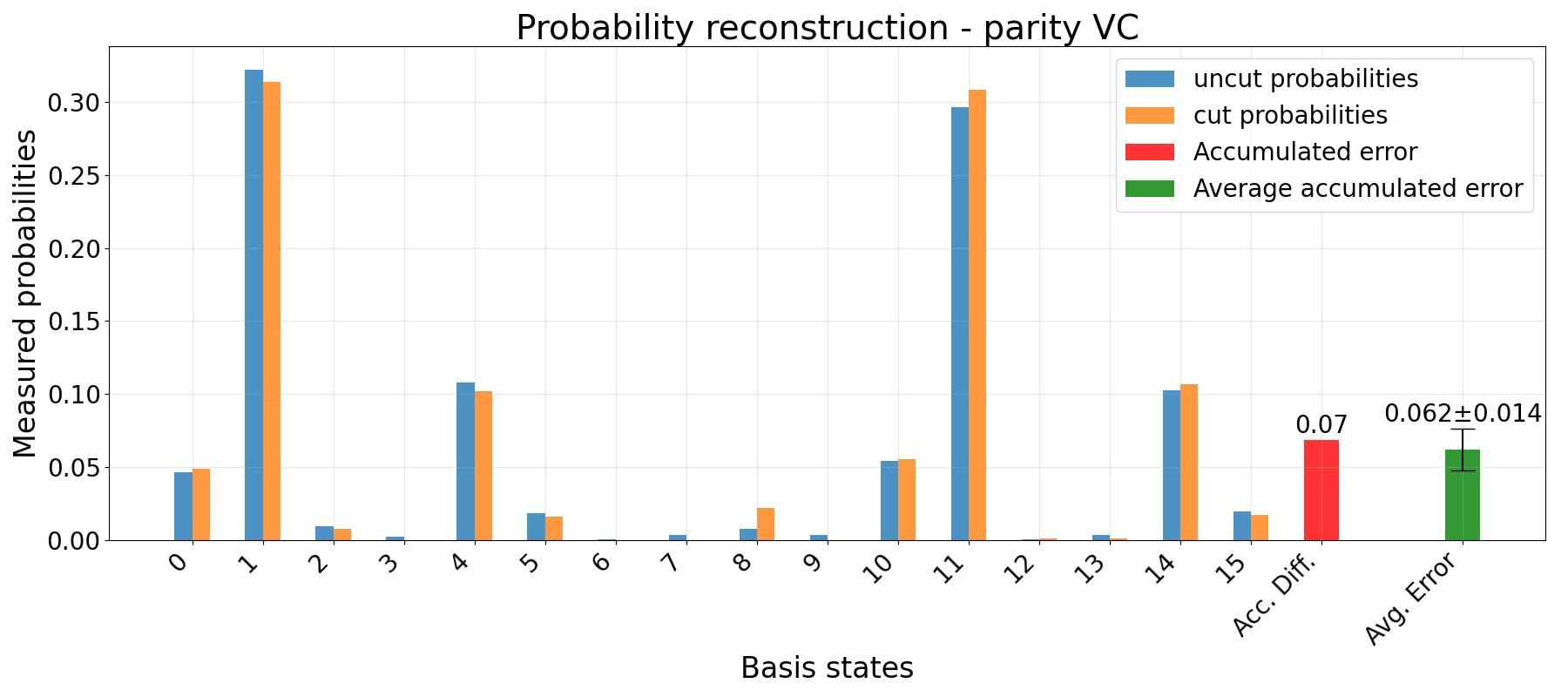}
        \caption{Probability reconstruction for parity classifier.}
    \label{fig:probability-reconstruction-VC-parity}
    \end{subfigure}
    \begin{subfigure}[b]{\linewidth}
        \centering
        \includegraphics[width=\linewidth]{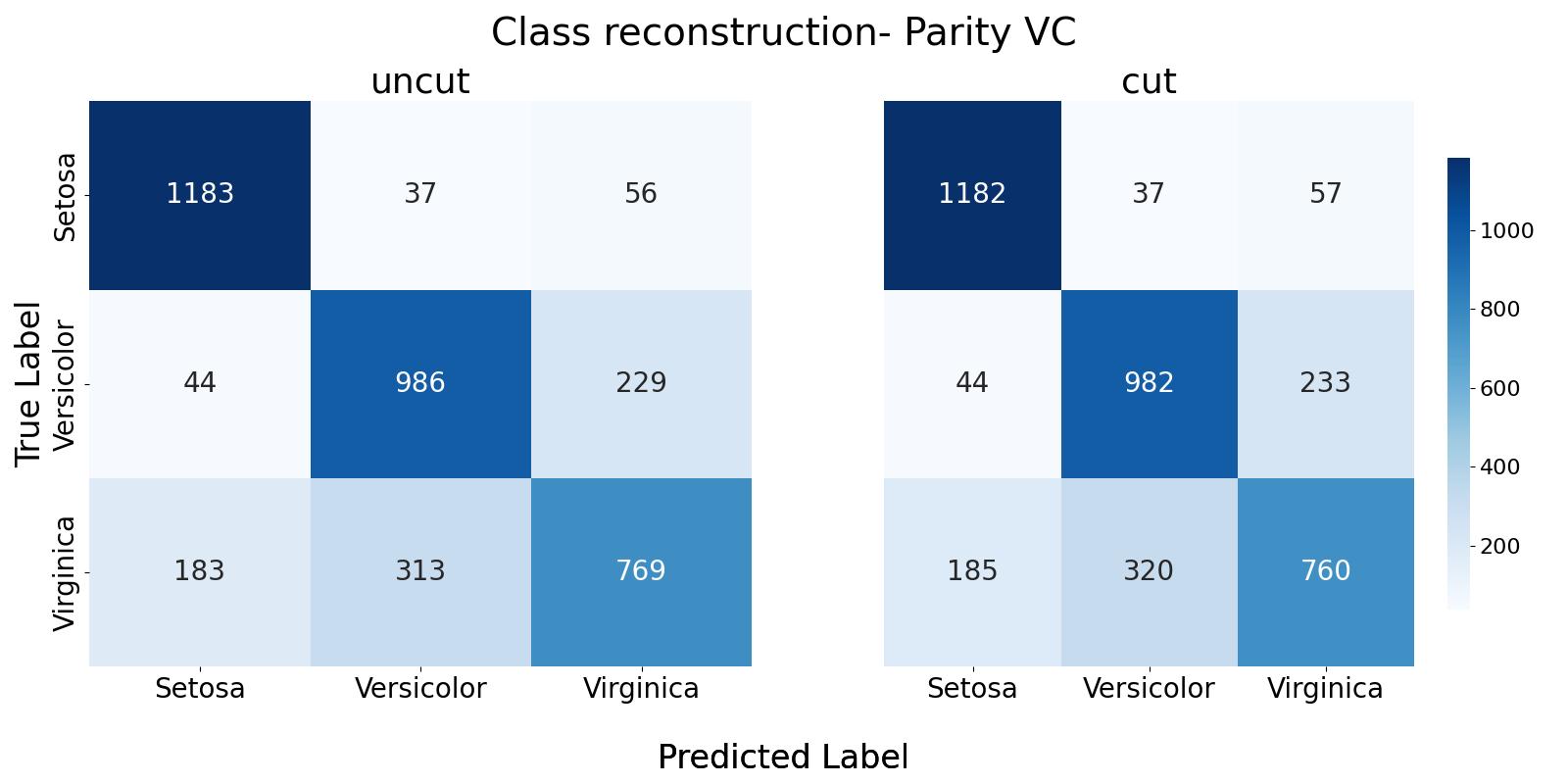}
        \caption{Aggregated results for 100 validations sets of 38 elements}
    \label{fig:class-reconstruction-PVC}
    \end{subfigure}
    \caption{Confusion matrices comparing the classifications produced by the cut and uncut version considering the parity VC.}
    \label{fig:class-reconstruction-PVC_all}
\end{figure}

\cref{fig:class-reconstruction-PVC_all} shows the aggregated confusion matrix for 100 different train-test sets.
We see that the results match those of the uncut circuit.
More interestingly, we observe the effect of a different mapping from measurement outcome to label: 
The \texttt{parity\_model} performs better on the two `difficult' classes, \emph{Versicolor} and \emph{Virginica}, than the \texttt{modulo\_model} and \texttt{expectation\_value\_model}.

\subsubsection{Cut then fit} \label{sec:parity_cut_and_fit}
\cref{fig:cut_and_fit_parity_model} shows the results for the \texttt{parity\_model} classifier employing the \texttt{cut\_then\_fit} strategy.
Once again, the scores calculated from 38 test samples show excellent agreement with those of the corresponding uncut model: score of 92.1\% with circuit cutting, compared to 84.2\% (after 11 epochs) and 86.8\% (after 10 epochs) for the uncut \texttt{VC} model.
Note that the model with this strategy has difficulty with correctly labeling \emph{Virginica} instances.
This fact is, however, also visible in the uncut model. 

The table below summarizes all obtained results for all discussed experiments. 
\begin{table}[htb]
\centering
\begin{threeparttable}
\caption{Summary of the final predicted scores on the Iris dataset for
all QML models trained in this work.}
\label{tab:model_scores}
\begin{tabular}{
    >{\centering\arraybackslash}m{3.5cm} 
    >{\centering\arraybackslash}m{2.0cm} 
    >{\centering\arraybackslash}m{2.0cm} 
    >{\centering\arraybackslash}m{2.0cm}
}
\toprule
\multirow{4}{*}{QML Model} & \multicolumn{3}{c}{Score (\%)} \\
\cmidrule(lr){2-4}
 & \multirow{2}{*}{\texttt{VC}\tnote{a}} & \multicolumn{2}{c}{\texttt{DQ VC}\tnote{a}} \\
 \cmidrule(lr){3-4}
 & & \texttt{fit\_then\_cut}\tnote{b} & \texttt{cut\_then\_fit}\tnote{d} \\
\midrule
\texttt{expected\_value\_model} & 66.6\tnote{b}\quad 73.7\tnote{c} & 66.4 & 78.9 \\
\texttt{modulo\_model} & 73.2\tnote{b}\quad 81.6\tnote{c} & 73.0 & 86.8 \\
\texttt{parity\_model} & 77.3\tnote{b}\quad 84.2\tnote{c} & 76.9 & 92.1 \\
\bottomrule
\end{tabular}
\begin{tablenotes}
\item[a]{\footnotesize \texttt{VC} = Original \texttt{VC} package with no circuit cut, \texttt{DQ VC} = modified \texttt{VC} package developed in this work to support circuit cutting.}
\item[b]{\footnotesize Values averaged over 100 tests with different splits in training and validation sets. Models trained for 120 iterations (1 epoch, batch size 2).}
\item[c]{\footnotesize Evaluated on 38 test samples. \texttt{VC} models trained for 55 iterations (11 epochs and batch size 25) with fixed random seed of 42.}
\item[d]{\footnotesize As in c. \texttt{DQ VC} models trained for 5 iterations (1 epoch and batch size 25) using previously optimized \texttt{VC} model parameters; see text.}
\end{tablenotes}
\end{threeparttable}
\end{table}

\subsection{Hardware and noisy simulation validation}
The final experiment focused on circuit cutting in presence of noise.
For this, we used both noisy simulations and hardware implementations. 
As we had only limited quantum hardware resources available, we only implemented the trained models for the \texttt{parity\_model} as this model showed the highest accuracy among the two models. 
We used the parameters from the noiseless training phase for the \texttt{fit\_then\_cut} training strategy. 
The noise model used for the simulations was derived from the noise parameters of the \texttt{ibm\_strasbourg} device, the device on which we implemented the quantum circuits. 

Using the exact simulation as our base truth, we define the error as the sum over all basis states of the absolute difference between the probability assigned to that basis state and the probability assigned to that basis state using exact simulation. 
\cref{fig:uncut_hardware_noisy} shows the results for the uncut quantum circuit. 
We averaged over 100 independent runs for the noisy simulation and the hardware implementation. 
The average error over these 100 runs is 0.35 and 0.614 for the noisy simulations and hardware experiments, respectively, with standard deviations of 0.014 and 0.168, respectively.

\cref{fig:cut_hardware_noisy} shows the results for the cut quantum circuits. 
The results for the noisy simulation are averaged over 50 independent runs, and due to resource constraints, only a single run using quantum hardware was performed. 
The error in the hardware results is 0.207, and the average error over 50 runs in the noisy simulation results is 0.143 with a standard deviation of 0.013.
These results show lower error rates for both noisy simulation and hardware implementation for the cut circuit than for the uncut circuit. 
\begin{figure}
    \centering
    \includegraphics[width=.8\linewidth]{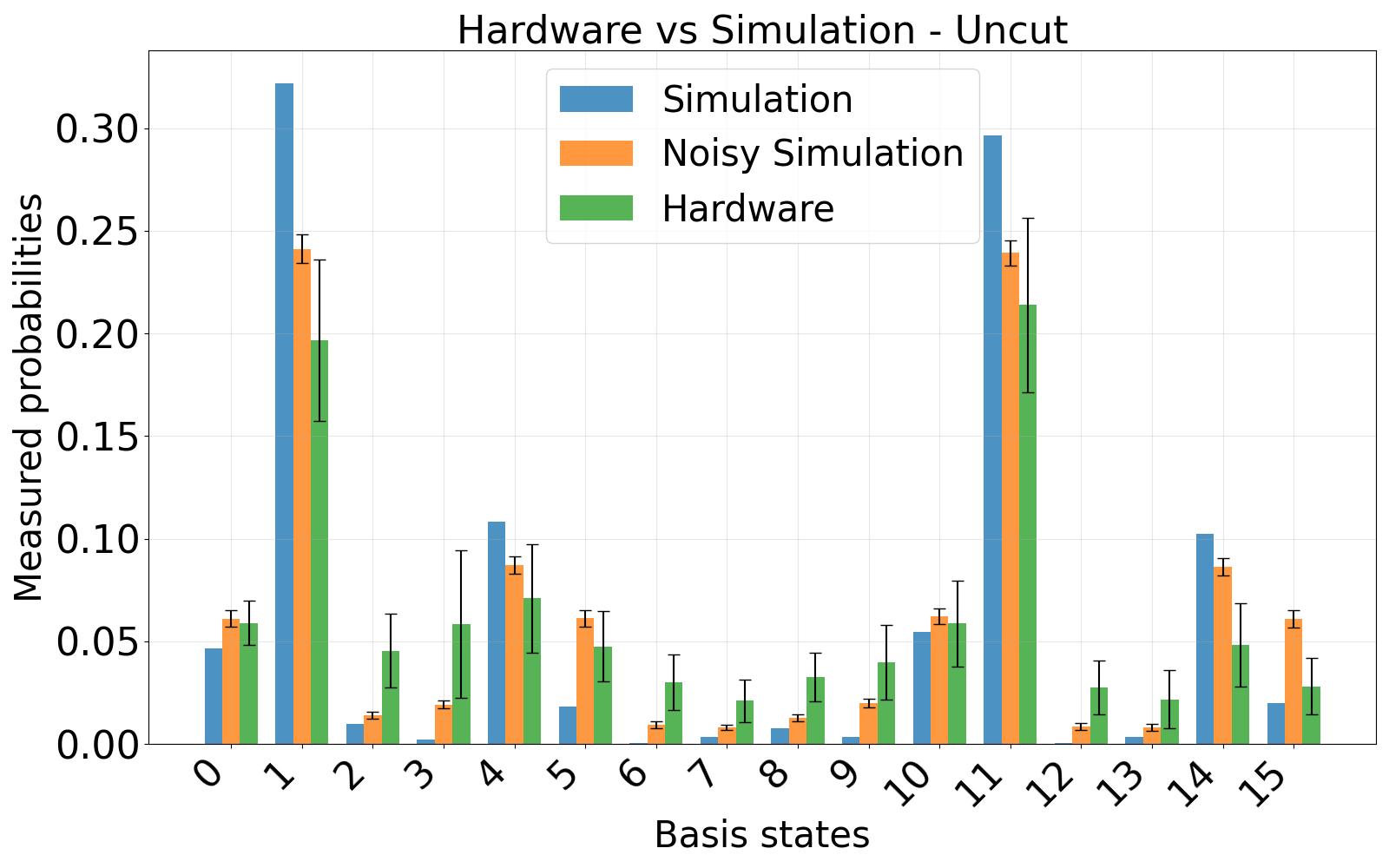}
    \caption{A comparison between probability distributions obtained from a single four-qubit circuit produced by the \texttt{parity\_model}, using exact simulation, and 100 repetitions using noisy simulation and hardware. The error bars show the standard deviation of the probability assigned to basis states within the 100 runs. The hardware and noise model used are from the IBM QPU \texttt{ibm\_strasbourg}.}
    \label{fig:uncut_hardware_noisy}
\end{figure}
\begin{figure}
    \centering
    \includegraphics[width=.8\linewidth]{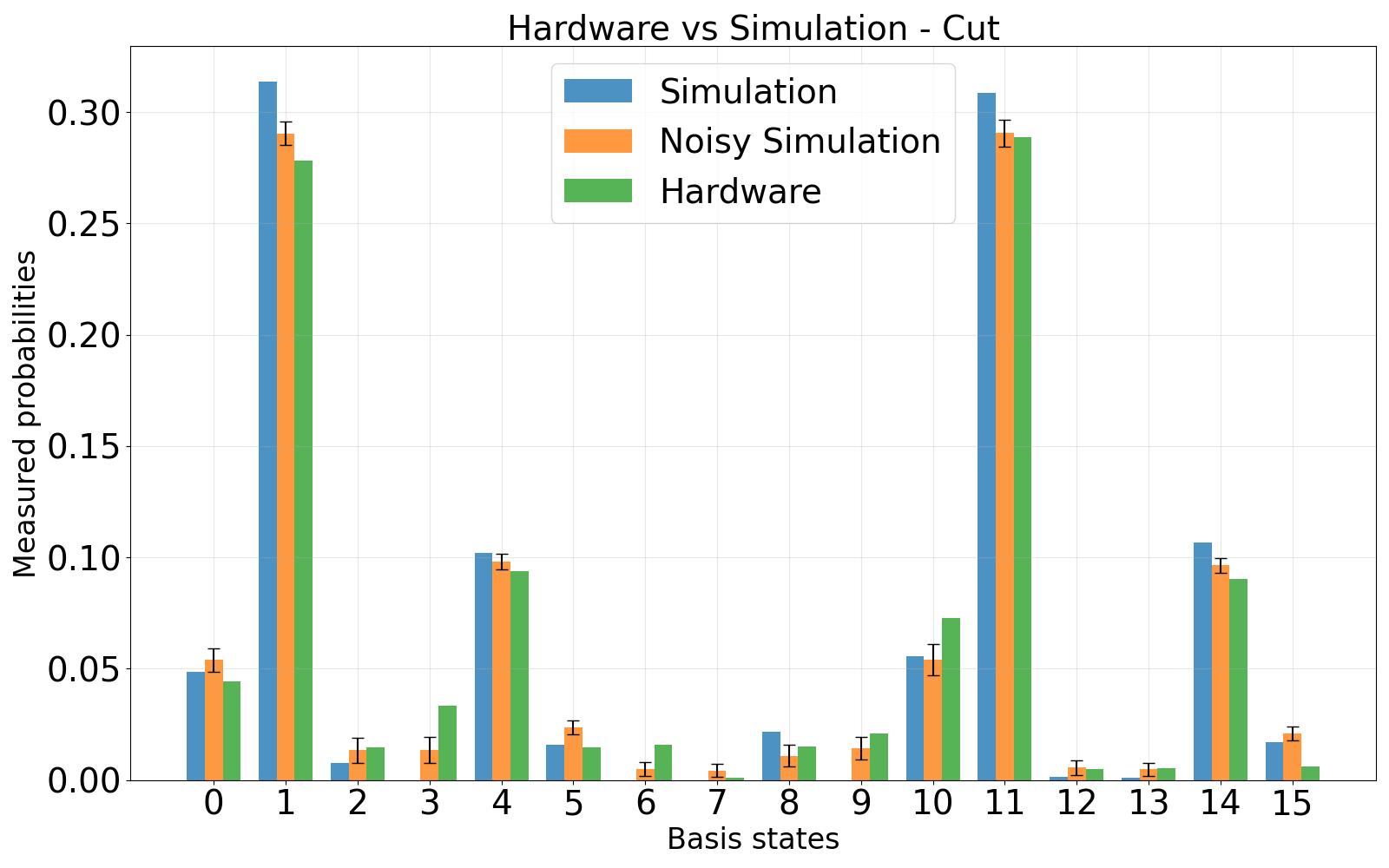}
    \caption{A comparison between probability distributions reconstructed after cutting a single circuit produced by the \texttt{parity\_model}, using exact simulation, hardware and 50 repetitions using noisy simulation. The error bars show the standard deviation of the probability assigned to basis states within the 50 runs. The hardware and noise model used are from the IBM QPU \texttt{ibm\_strasbourg}.}
    \label{fig:cut_hardware_noisy}
\end{figure}

\section{Discussion}\label{sec:Discussion}
This work introduced a new framework for conducting circuit cutting in quantum circuits. 
We empirically validated the correctness and then showed its potential using different use cases. 
Additionally, we introduced a new method to reconstruct the probability distribution of a quantum circuit, a notoriously hard task in quantum computing in general and circuit cutting specifically. 

Circuit cutting techniques can help to significantly reduce circuit depth and circuit size, making them (more) suitable to run on quantum hardware. 
These methods do however come with classical overhead:
A larger quantum circuit is decomposed in multiple smaller quantum circuits, for instance by replacing multi-qubit gates by single-qubit gates and measurements. 
Each of these smaller quantum circuits can then independently be run on quantum hardware, and the results are recombined using a classical computer. 
Noisy hardware can especially benefit from this technique, as it can lower the quantum resource requirements. 

Based on the validations and benchmarks performed using both random circuits and variational classifier models, we demonstrate that the proposed probability distribution reconstruction approach is highly robust and performs similar to expectation-value reconstruction-based methods. 
Despite the higher number of possible outcomes, the method still obtained a similar average accumulated error as the expectation value outcomes. 
We note that this method is useful for many practical quantum algorithms, where a probability distribution is prepared with a few states with high amplitude. 
In those cases, the error margin of the method suffices to still obtain correct results. 
For algorithms that have a approximately flat final probability distribution, this method will not work, as in those cases the individual probabilities are exponentially small. 

Next, we considered two training strategies for quantum machine learning applications, \texttt{fit\_then\_cut} and \texttt{cut\_then\_fit}. 
The former is easiest, as circuit cutting is only required once, to run the circuit with the optimal parameters. 
The second strategy is most realistic, also during training, limited resources are available. 
Due to the high computational cost of the latter method, we initialized the training with already trained results from an uncut model, and then trained for 5 more iterations in a cut model.
Future work can explore the effect of doing the entire training in a cut-like manner. 

Simulating a quantum circuit is hard, especially as the circuit grow in size. 
Additionally, the circuit cutting techniques give multiple quantum circuits that have to be run independently on quantum hardware. 
Finally, training a quantum circuit requires multiple evaluations to update the model parameters. 
These three effects combined make that simulating circuit cutting techniques for variational classifiers is hard, especially when employing the \texttt{cut\_then\_fit} strategy. 
The simulation time for a single experiment could be up to one week, forcing us to use only a limited number of experiments for some results. 
Still, we observe that the \texttt{cut\_then\_fit} strategy works and produces similar results as the uncut model and when using the \texttt{fit\_then\_cut} strategy.

As an added benefit, this work empirically shows the potential of different encoding methods. 
The \texttt{parity\_model}, despite its more complex mapping between quantum states and class labels, shows higher classification scores than the other mappings. 

All experiments showed incorrect assignments after training, especially with the \emph{Versicolor} and \emph{Virginica} classes. 
However, these errors appear in both the cut and the uncut experiments, indicating a cause independent of the applied circuit cutting methods. 
Most likely, the errors originate from the limited flexibility of the overall QML model architecture.
Standard methods, such as kernels~\cite{Boser:1992}, might help achieve better accuracy, posing an interesting direction for future research. 

Looking at the results using hardware and noisy simulations, we see that circuit cutting mitigates the impact of noise significantly. 
In addition, for the circuit cutting setting, the results for the hardware implementation are well approximated by both the noisy simulation and even the noiseless simulation, opposed to the uncut circuit. 
Note that interestingly, cut quantum circuits perform better in this case than uncut quantum circuits, overcoming the additional stochasticity introduced by the circuit cutting methods. 
This observation highlights the potential of circuit cutting for large industry use cases.  

Due to limited quantum resources, the quantum hardware results cover only a small part of the experiments carried out for the noiseless simulation. 
Extra quantum resources would allow for a thorough statistical analysis of the significance of the results. 
Additionally, full access to the quantum hardware, in contrast to the black-box approach employed in this work, allows for fine-tuning of the implementations to obtain better results. 

Future work can furthermore focus on larger quantum circuits and algorithms, and on other kinds of quantum algorithms, both variational and algorithms that do not require training. 
Other directions could focus on the quantum variational classifier and test other mixing layer layouts or other training methods.

\section*{Acknowledgments}
This publication is part of the project \textit{Divide and Quantum} `D\&Q' NWA.1389.20.241 of the program `NWA-ORC', which is partly funded by the Dutch Research Council (NWO). 

\printbibliography

\end{document}